\documentclass[lettersize,journal]{IEEEtran}

\usepackage{epsfig,amsmath,amssymb,epsf,amsthm,scalefnt,multirow}
\usepackage{xcolor}
\usepackage{float}
\usepackage[caption=false,font=footnotesize,textfont=sf]{subfig}
\usepackage{cite}
\usepackage{psfrag}
\usepackage{bm}
\usepackage{algorithmic,algorithm}
\usepackage{graphicx}
\usepackage{mleftright}

\DeclareGraphicsExtensions{.pdf,.png,.jpg}

\newtheorem{theorem}{Theorem}
\newtheorem{lemma}{Lemma}
\newtheorem*{lemma*}{Lemma}

\begin{document}
\title{BanditLinQ: A Scalable Link Scheduling for
\\Dense D2D Networks with One-Bit Feedback}

\author{Daeun~Kim,~\IEEEmembership{Student~Member,~IEEE},~and~Namyoon~Lee,~\IEEEmembership{Senior~Member,~IEEE}

 \thanks{This work was supported in part by the National Research Foundation of Korea (NRF) Grant through the Ministry of Science and ICT (MSIT), Korea Government, under Grant 2022R1A5A1027646 and in part by the National Research Foundation of Korea (NRF) Grant funded by the Korea Government (MSIT), under Grant RS-2023-00208552. }
\thanks{Daeun Kim is with the Department of Electrical Engineering, Pohang University of Science and Technology (POSTECH), Pohang 37673, South Korea  (e-mail: daeun.kim@postech.ac.kr).}
	\thanks{Namyoon Lee is with the School of Electrical Engineering, Korea University, Seoul 02841, South Korea (e-mail: namyoon@korea.ac.kr).}

}

  \maketitle

\begin{abstract}

  This paper addresses cooperative link scheduling problems for base station (BS) aided device-to-device (D2D) communications using limited channel state information (CSI) at BS. We first derive the analytical form of ergodic sum-spectral efficiency as a function of network parameters, assuming statistical CSI at the BS. However, the optimal link scheduling, which maximizes the ergodic sum-spectral efficiency, becomes computationally infeasible when network density increases. To overcome this challenge, we present a low-complexity link scheduling algorithm that divides the D2D network into sub-networks and identifies the optimal link scheduling strategy per sub-network. Furthermore, we consider the scenario when the statistical CSI is not available to the BS. In such cases, we propose a quasi-optimal scalable link scheduling algorithm that utilizes one-bit feedback information from D2D receivers. The algorithm clusters the links and applies the UCB algorithm per cluster using the collected one-bit feedback information. We highlight that even with reduced scheduling complexity, the proposed algorithm identifies a link scheduling action that ensures optimality within a constant throughput gap. We also demonstrate through simulations that the proposed algorithm achieves higher sum-spectral efficiency than the existing link scheduling algorithms, even without explicit CSI or network parameters knowledge.

\end{abstract}


\section{Introduction}

\IEEEPARstart{T}{he} exponential growth of Internet of Things (IoT) devices, along with various communication types, has increased the demand for high-speed connectivity and data rates \cite{akpakwu2017}. Additionally, the low power consumption, low cost, and security requirements of IoT devices require novel techniques. To address these demands, Low Power Wide Area (LPWA) technologies, such as Narrowband IoT (NB-IoT), have received significant attention \cite{akpakwu2017,rastogi2020,migabo2020}. NB-IoT, as a cellular LPWA technology, is expected to reduce deployment costs by allowing for deployment within existing infrastructure and available spectrum \cite{rico2016}. However, the poor quality of the link between NB-IoT user equipment and the base station (BS) can lead to issues such as low reliability and limited coverage \cite{gbadamosi2022}. A potential solution to address these challenges is device-to-device (D2D) communication, which enables direct communication between devices \cite{orsino2017, Tehrani2014}. D2D communication can reduce the traffic and computing burden of the BS, improve spectral efficiency, and network coverage by utilizing better direct links between devices instead of transmitting via the BS \cite{Feng2014,Zhang2017}.

The effective spectrum sharing of the D2D communication networks has become increasingly important due to the growing number of IoT devices, which can enhance throughput and coverage compared to traditional cellular networks \cite{Huynh2016}. However, this task presents a significant challenge for interference management. Fortunately, it is possible to use centralized interference management approaches thanks to the capability of the BS coordination in NB-IoT networks. To mitigate interference optimally in a centralized manner, the BS should obtain global channel state information (CSI) from dense D2D networks, including the D2D network topology and channel realizations of all links. Nonetheless, acquiring such CSI accurately and timely is infeasible in practice by the limited capacity of the feedback channel. In this paper, we propose efficient D2D link scheduling techniques using limited CSI under a centralized BS coordination for dense D2D communication networks. In particular, we shall show that it is possible to increase the throughput of D2D networks as increasing the density of D2D links with one-bit feedback information per D2D link.

\subsection{Related Works}

Interference management techniques for general interference networks have been proposed in \cite{Jafar2011} which require global and instantaneous channel state information at transmitter (CSIT). Especially for D2D networks, considerable research has focused on managing interference through mode selection, power control, and resource allocation. Many centralized power control or resource allocation techniques, such as those in \cite{Lee2015,Yin2016,Maghsudi2015,Yu2011,Jiang2018,Naqvi2018,Takshi2018,Khuntia2019,Gao2019}, relied on global CSI. In \cite{Maghsudi2015}, game and graph theory were used to handle joint channel allocation and power control problems. Resource allocation and power control methods that maximize system throughput were proposed in \cite{Jiang2018,Naqvi2018,Gao2019}. However, in dense networks, these centralized methods required extensive and instantaneous signaling exchanges to acquire accurate and timely CSI. 

To address this issue, interference management and resource allocation techniques that rely on imperfect CSIT were proposed in \cite{Tse2012,Lee2014,Jafar2014,Naderializadeh2015,Doumiati2019,Wang2017,Li2020,Xiao2020}. Specifically, interference alignment methods were proposed under the assumption of delayed CSIT \cite{Tse2012,Lee2014}. Further, the topological interference management (TIM) techniques were proposed in \cite{Jafar2014,Naderializadeh2015,Doumiati2019}, which use imperfect CSIT comprised of the connectivity information. However, the optimization problem for solving the TIM problem is usually NP-hard due to the non-convex objective function. In \cite{Wang2017,Li2020,Xiao2020}, the resource allocation problem with statistical CSI was addressed for D2D communications. However, the statistical CSI, which is obtained with long-term channel observation and statistical fitting, becomes unavailable if the environment changes.

Distributed interference mitigation techniques have also been considered as effective approaches for enhancing the throughput of D2D communication networks. Such distributed techniques require relatively less CSI compared to the centralized methods \cite{Yin2016,Geng2015,Wu2013,Navid2014,Lee2015,Lee2015_2}.
For instance, distributed link scheduling algorithm based on the interference level to and from each others' links was proposed in \cite{Wu2013} as an effective link scheduling solution for a large number of links. To improve the sum-spectral efficiency, information-theoretically optimal D2D link scheduling method was proposed in \cite{Navid2014}, which leverages the optimality condition of treating interference as noise in interference networks \cite{Geng2015}. Further, in \cite{Lee2015}, an optimal distributed on-off power control method was presented, which maximizes the sum rate of D2D links. In addition, advanced interference cancellation techniques were used to improve the spectral efficiency of D2D communication networks \cite{Lee2016}. However, these distributed approaches cannot fully harness the capability of the BS coordination in NB-IoT networks.



Machine learning (ML) has emerged as a promising solution for interference management problems that don't necessarily require explicit channel state information (CSI) or network information. One particular framework that has gained considerable attention for real-time interference management is the multi-armed bandit (MAB) framework, which is a type of reinforcement learning (RL). Several studies \cite{Begashaw2016, Zhang2016, Kim2020, Kim2019} have proposed MAB-based algorithms for interference management. For instance, \cite{Begashaw2016} proposed a blind interference alignment algorithm with reconfigurable antennas that uses the upper confidence bound (UCB) algorithm \cite{Auer2002} for antenna mode selection. Similarly, \cite{Zhang2016} proposed a distributed beamforming approach for multicell interference networks that adjusts the signal-to-interference-plus-noise ratio (SINR) targets at users using the UCB algorithm. Although \cite{Kim2020} introduced a UCB-based multi-user precoding strategy for frequency-division duplexing (FDD) massive multiple-input multiple-output (MIMO) downlink systems, the algorithm's fast-exploration strategies do not guarantee optimal performance. Furthermore, \cite{Kim2019} proposed a MAB-based transmission coordination algorithm for heterogeneous ultra-dense networks that reduces computational overhead by introducing a new interference metric and optimizing it. Although these studies have proposed methods for interference management without global CSI, they still suffer from computational complexity issues and have not considered scalable centralized interference management with one-bit feedback information.

The resource allocation challenges in D2D networks have been tackled using MAB based frameworks in various studies \cite{Hashima2022, Ortiz2019,Gai2010, hakami2022}. In \cite{Hashima2022}, a distributed subband-selection algorithm was proposed, which reduces computational complexity and does not require complete network information. To solve combinatorial problems efficiently, combinatorial MAB (CMAB) based resource allocation approaches were suggested in \cite{Ortiz2019, Gai2010, hakami2022}. In \cite{Ortiz2019}, the CMAB approach was used to deal with the combinatorial nature of mode selection and resource allocation problems, leveraging the naive sampling strategy (NS) in \cite{Ontan2013}. However, the selection of the super-arm leads to exponential growth of computational complexity with the number of D2D links. In \cite{Gai2010}, a matching-learning algorithm was proposed to address the multi-user channel allocation issue in cognitive radio networks. Similarly, \cite{hakami2022} presented a matching learning-based algorithm for D2D resource allocation without CSI, but the algorithm is expensive in each iteration. Despite these concerns, the ML-based approaches provide scalable solutions to the D2D resource allocation issue without global CSI.

\subsection{Contributions}
In this paper, we consider D2D link scheduling problem for D2D communications in ultra-dense NB-IoT networks, in which $K$ D2D transmitter and receiver pair communicate directly in the common spectrum. Our link scheduling problem focuses on maximizing the ergodic sum-spectral efficiency and sum-throughput using statistical CSI and one-bit CSI at the BS. The major contributions of this paper are summarized as follows.

\begin{itemize}
  \item  We begin by deriving an analytical expression for the ergodic sum-spectral efficiency of D2D networks when the BS has knowledge of statistical CSI, i.e., both the path-loss and the fading channel distributions of all D2D links. The ergodic sum-spectral efficiency is characterized as a function of several network parameters, including the path-loss exponent, link distances, fading parameters,  and link scheduling action. Leveraging the derived ergodic sum-spectral efficiency, we propose an optimal link scheduling action that maximizes the ergodic sum-spectral efficiency by solving a combinatorial optimization problem. However, finding the optimal scheduling action becomes computationally complex with the increasing number of D2D links, making it impractical for dense D2D networks. To address this issue, we present a low-complexity algorithm based on D2D clustering, which partitions the D2D network into sub-networks using a modified hierarchical clustering algorithm. Our clustering method selects D2D communication pairs that strongly interfere with each other while maintaining negligible inter-cluster interference. We then identify the optimal D2D link scheduling action per cluster that maximizes the ergodic sum-spectral efficiency averaged over inter-cluster interference distributions.

\item Next, we consider the case where the BS has only knowledge of ACK/NACK feedback information per D2D link, i.e., one-bit feedback. In this case, we present a quasi-optimal D2D link scheduling algorithm that maximizes the network sum-throughput. The proposed method consists of two phases: D2D network clustering and upper confidence bound (UCB) based link scheduling per cluster. We refer to this as \textit{BanditLinQ} algorithm.  In the initial phase, the BS forms multiple D2D sub-networks using the one-bit feedback information from the D2D receiver. In the second phase, BS adaptively selects the D2D link scheduling action per cluster to maximize the empirically computed average throughput. We prove that the proposed algorithm finds quasi-optimal cooperative scheduling action that ensures optimality within a constant-gap from the throughput attained by the optimal cooperative action. As a result, our algorithm is scalable to find the optimal scheduling action in a dense D2D network within a constant gap while reducing the scheduling complexity exponentially with the number of clusters. 

\item We compare the proposed D2D link scheduling algorithm with existing methods. Our simulation results demonstrate that the proposed D2D link scheduling algorithm outperforms existing algorithms in average sum-spectral efficiency. We verify that the gain is more pronounced as the density of D2D links increases.

\end{itemize}

\section{System Model and Problem Statement}
In this section, we first explain the BS assisted D2D communication networks and the channel models. Then, we define the D2D link scheduling problems according to the knowledge levels of CSI at the BS.

\subsection{Network and Channel Model} We consider a BS-aided D2D communication network, where $K$ D2D transmitters and their corresponding receivers communicate directly, sharing the same frequency and time resources. The transmitters are uniformly distributed within a cell. All D2D transmitters and receivers are equipped with a single antenna. We assume that the link scheduling of the D2D pairs is globally controlled by the BS, which sends scheduling commands to the D2D transmitters. We also assume that the BS serves the cellular users using orthogonal time-frequency resources to the D2D communications, i.e., no interference between cellular and D2D links. This assumption ensures efficient and reliable cellular and D2D communications as in NB-IoT applications  \cite{akpakwu2017,rastogi2020,migabo2020}.

We define the path-loss model from the $\ell$th transmitter to the $k$th receiver as $d_{k,\ell}^{-\frac{\beta_{k,\ell}}{2}}$, where $d_{k,\ell}$ denotes the distance between them and $\beta_{k,\ell}$ is a path-loss exponent. Additionally, we use $h_{k,\ell}[t]$ to represent the complex channel coefficient from the $\ell$th transmitter to the $k$th receiver at fading block $t$.

We denote the fading from D2D transmitter $k$ to receiver $k$ by $h_{k,k}[t]$. To model the randomness of the small-scale channel fading process of the desired link, we use the Nakagami-$m$ distribution as
\begin{align}
    f_{h_{k,k}[t]}(x)=\frac{2m^m}{\Gamma(m)}x^{2m-1}\exp\left(-mx^2\right),
\end{align}
where $\Gamma(m)=(m-1)!$ for every positive integer $m$. This distribution is chosen to capture the line-of-sight (LOS) effects, and it allows us to select an appropriate value for the parameter $m$. Unlike the desired link, we model the channel fading for the interfering links, i.e., $h_{k,\ell}[t]$ for $k\neq \ell$, as the Rayleigh distribution, which is a special case of the  Nakagami-$m$ when $m=1$. We assume that a block fading process, where the channel coefficient $h_{k,\ell}[t]$ changes independently in every fading coherence time duration $T_c$, while it does not change within $T_c$. In addition, the path-loss $d_{k,\ell}^{-\frac{\beta_{k,\ell}}{2}}$, i.e., the network topology, remains constant over multiple fading blocks.

Using the defined network and channel models, the received signal at the D2D receiver $k$ in fading block $t$ can be expressed as
\begin{align}
    y_k[t] \!=\! h_{k,k}[t]d_{k,k}^{-\frac{\beta_{k,k}}{2}}s_k[t] \!+\! \sum_{\ell=1, \ell \ne k}^{K} h_{k,\ell}[t]d_{k,\ell}^{-\frac{\beta_{k,\ell}}{2}}s_{\ell}[t] + n_k[t], 
\end{align}
where $s_k[t]$ is the transmit symbol of transmitter $k$, which is assumed to be drawn from $\mathcal{CN}(0,P)$ to meet the average power constraint $\mathbb{E}[|s_k[t]|^2] = P$, and $n_k[t]$ denotes the additive Gaussian noise, i.e., $n_k[t] \sim \mathcal{CN}(0,\sigma^2)$. 

\begin{figure*}[!t]
\begin{center}
\includegraphics[width=13cm]{ 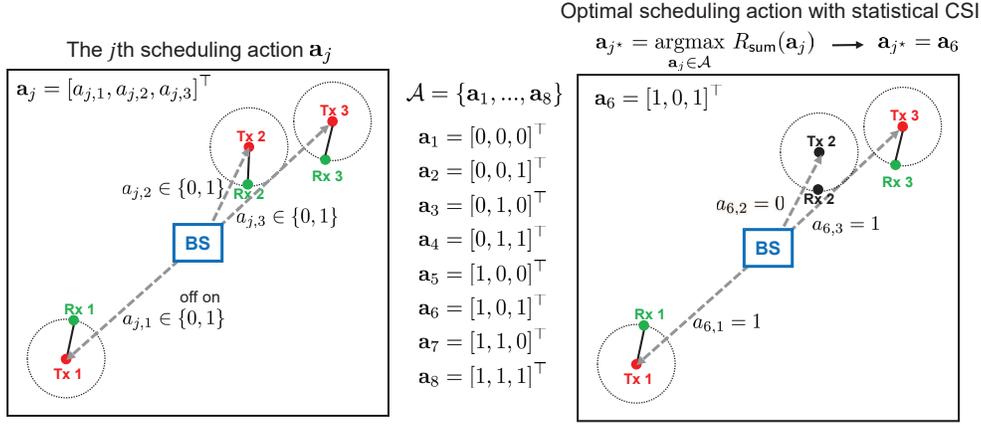}
\end{center}
\caption{An illustration of the link scheduling actions and optimal cooperative link scheduling action when $K=3$.}\label{fig:optimal}
\end{figure*}

\subsection{Cooperative D2D Link Scheduling Problem }
We explain cooperative D2D link scheduling problems with different CSI assumptions at the BS. We assume that a BS collects global information from $K$ D2D receivers to determine which links to activate by selecting the
scheduling action for each link from the on-off set $\{0, 1\}$. Using the global information, the BS selects a cooperative scheduling action to maximize utility functions. 

Since the BS selects scheduling action for each link from the on-off set $\{0, 1\}$, a total of $J=2^K$ possible cooperative scheduling actions exist. We define the $j$th cooperative scheduling action vector as ${\bf a}_j=[a_{j,1},\ldots, a_{j,K}]^{\top}\in\{0,1\}^{K}$, where $a_{j,k}\in\{0,1\}$ is the scheduling action of the $k$th link when the cooperative scheduling action is ${\bf a}_j$. By collecting all possible scheduling actions, we define the scheduling action space as
 \begin{align}
     \mathcal{A}=\{{\bf a}_1,{\bf a}_2,\ldots,{\bf a}_J\}.
 \end{align}
The BS needs to identify the optimal scheduling action to maximize network utility functions, which are based on the available CSI knowledge levels at the BS.

\subsubsection{Instantaneous Sum-Spectral Efficiency Maximization} 
We begin by assuming that the BS has access to perfect knowledge of the distances, path-loss exponents, and the channel realization $h_{k,\ell}[t]$ of all links for every fading block. With this information, the BS is able to compute the instantaneous rate $R_{k}[t]({\bf a}_j)$ for the $k$th link, taking into account the cooperative scheduling action $\mathbf{a}_j$, i.e.,
\begin{align}
    R_{k}[t]({\bf a}_j) \! =\! \log_2 \!\left(\!
 1\!+\! \frac{\left|h_{k,k}[t]\right|^2d_{k,k}^{-\beta_{k,k}}{a}_{j,k} }{\sum_{\ell\neq k}\left|{h}_{k,\ell}[t]\right|^2 d_{k,\ell}^{-\beta_{k,\ell}}{a}_{j,\ell} +\frac{1}{\sf snr}} \!\right), \label{eq:SE1}
\end{align}
where ${\sf snr}=\frac{P}{\sigma^2}$ is the signal-to-noise ratio (SNR) of the system.
From \eqref{eq:SE1}, we define the instantaneous sum-spectral efficiency as a utility function: 
\begin{align}
    R_{\sf sum}[t] ({\bf a}_j) =  \sum_{k=1}^KR_{k}[t]({\bf a}_j). \label{eq:ISE}
\end{align}
This rate calculation is crucial in determining the optimal scheduling strategy for maximizing the instantaneous sum-spectral efficiency as 
\begin{align}
    \mathbf{a}_{j^\star}[t]  = \underset{\mathbf{a}_j \in \mathcal{A}}{\mathrm{argmax}}~  R_{\sf sum}[t] ({\bf a}_j).
\end{align}
This centralized D2D link scheduling problem has been extensively studied in literature \cite{Lee2015,Yin2016,Maghsudi2015,Yu2011,Jiang2018,Naqvi2018,Takshi2018,Khuntia2019,Gao2019}. Unfortunately, acquiring perfect and timely CSI at the BS is infeasible in practice, we shall not focus on this problem in this paper.

\subsubsection{Ergodic Spectral Efficiency Maximization with Statistical CSI} 

The BS can leverage the long-term CSI of the D2D network, such as the path-loss $d_{k,\ell}^{-\frac{\beta_{k,\ell}}{2}}$. However, obtaining accurate and timely realizations of fadings $h_{k,\ell}[t]$ can be challenging due to the mobility of D2D transceivers and the limited feedback channel capacity. Instead, the BS can exploit information on the statistical distributions of $h_{k,\ell}[t]$ to maximize a network utility function. Using this information, the BS can calculate the network utility function in terms of the ergodic spectral efficiency, which is computed by taking the average on the instantaneous spectral efficiency in \eqref{eq:ISE} over the fading distributions, i.e.,
\begin{align}
    & R_{k}(\!{\bf a}_j\!)\!=\!\mathbb{E}_{\left|{h}_{k,\ell}[t]\right|^2} \!\!\left[\! \log_2\!\!\left(\!
 1\!+\! \frac{\left|h_{k,k}[t]\right|^2d_{k,k}^{-\beta_{k,k}}{a}_{j,k} }{\sum_{\ell\neq k}\!\left|{h}_{k,\ell}[t]\right|^2 \! d_{k,\ell}^{-\beta_{k,\ell}}{a}_{j,\ell}\! +\!\frac{1}{\sf snr}} \!\right)\!\right].
\end{align}
As a result, the ergodic sum-spectral efficiency for the cooperative scheduling action $\mathbf{a}_j$ is given by
\begin{align}
     R_{\sf sum}({\bf a}_j) =  \sum_{k=1}^K R_{k}({\bf a}_j).
\end{align}
The optimal cooperative scheduling action $\mathbf{a}_{j}$ is found by solving the following optimization problem: 
\begin{align}
    \mathbf{a}_{j^\star}  = \underset{\mathbf{a}_j \in \mathcal{A}}{\mathrm{argmax}} ~R_{\sf sum}({\bf a}_j). \label{eq:ESE_opt}
\end{align}
{As an example of $K = 3$, the link scheduling actions and optimal link cooperative scheduling action are illustrated in Fig. \ref{fig:optimal}.}
One challenge to solving this optimization problem is that we need to compute the ergodic sum-spectral efficiency in an analytical form. Additional challenge is that \eqref{eq:ESE_opt} is a combinatorial optimization problem, implying that the computational complexity increases with the number of the D2D pairs $K$. As a result, we need to develop a computational-efficiency algorithm that finds quasi-optimal scheduling action, which will be explained in Section III. 


\subsubsection{Throughput Maximization with One-Bit Feedback}
The scenario in which the BS has no knowledge of the long- and short-term fadings due to the high mobility of the D2D transceivers is of practical relevance. In  this scenario, the BS can obtain one-bit feedback information from D2D receivers. Specifically, each D2D receiver sends an acknowledgment (ACK) signal to the BS if its instantaneous spectral efficiency with scheduling action ${\bf a}_j$ meets the target rate $r_k$, i.e., $R_{k}[t]({\bf a}_j)>r_k$. Otherwise, it sends a negative-acknowledgment (NACK) signal.

Using this feedback information, the BS can obtain the instantaneous sum-throughput ${\bar R}_{\sf sum}[t]({\bf a}_j)$. The sum-throughput is a measure of the total data rate that can be achieved by all the D2D links that are active in the current time slot $t$. 
The instantaneous sum-throughput is computed by summing the spectral efficiencies of all the active D2D links that have sent an ACK signal as
\begin{align}
    {\bar R}_{\sf sum}[t]({\bf a}_j) = \sum_{k=1}^{K} r_k \mathbf{1}\{R_{k}[t]({\bf a}_j)>r_k\}, \label{eq:one_bit_uility}
\end{align}
where ${\bf 1}_{\{{E}\}}$ is an indicator function that yields one if the event set $E$ is true; otherwise, it is zero. This enables the BS to determine the overall performance of the D2D network and adjust its scheduling actions accordingly. Therefore, the goal is to find the sequence of the optimal cooperative scheduling actions $\pi[t] \in \{1,2,\ldots,J\}$ that maximizes the empirically averaged sum-throughput during the finite $T$ fading blocks as follows:
\begin{align}
    \left( {\pi}^\star[1],\ldots, {\pi}^\star[T] \right)= \underset{\pi[1],\pi[2],\ldots, \pi[T]\in [J]^T}{\mathrm{argmax}} \frac{1}{T}\sum_{t=1}^{T} {\bar R}_{\sf sum}[t]({\bf a}_{\pi[t]}).
\end{align}
The task of determining the optimal sequence of cooperative scheduling actions presents a significant challenge compared to solving the optimization problem presented in \eqref{eq:ESE_opt}. One of the reasons for this is that the BS lacks knowledge of the utility function for the given scheduling action, i.e., ${\bar R}_{\sf sum}[t]({\bf a}_j)$. Instead, the BS must rely on estimating the sum-throughput using only one-bit ACK/NACK feedback provided by the receivers. This estimation method further adds to the complexity of the optimization problem. Moreover, finding a computationally-efficient algorithm that determines the sequence of cooperative scheduling actions is also necessary. To address this challenge, in Section IV, we introduce a quasi-optimal algorithm based on the MAB framework. 




\section{Optimal Link Scheduling with Statistical CSI} \label{sec:opt}
In this section, we propose the optimal cooperative scheduling action that maximizes ergodic sum-spectral efficiency when the BS has the statistical CSI.

We first introduce a lemma that helps to compute the ergodic sum-spectral efficiency in an analytical form.

\begin{lemma}
Let $X$ and $Y_i$ for $i=1,\ldots, N$ be $N+1$ independent and non-negative random variables. Then,
\begin{align}
&\mathbb{E}\left[\ln \left(1+\frac{X}{\sum_{i=1}^N Y_i+ \frac{1}{{\sf snr}}}\right)\right] \nonumber\\&=\int_{0}^{\infty}\frac{e^{-\frac{z}{{\sf snr}}}}{z} \left(1-\mathbb{E}\left[e^{-zX}\right]\right) \prod_{i=1}^N\mathbb{E}\left[e^{-zY_i}\right] {\rm d}z.
\end{align}
\end{lemma}
\begin{proof}
    See \cite{Hamdi2010}.
\end{proof}

Utilizing Lemma 1, we can establish the analytical expression for the ergodic sum-spectral efficiency, which is stated in the following theorem.
\begin{theorem}
The ergodic sum-spectral efficiency with scheduling action ${\bf a}_j \in \mathcal{A}$ is given by  
\begin{align}
&R_{\sf sum} ({\bf a}_j) =\log_2{e}\sum_{k=1}^K  a_{j,k} \int_{0}^{\infty}\!\!\frac{e^{-\frac{z}{{\sf snr}}}}{z} \!\left(\!\!1-\!\frac{1}{\left(1+\frac{zd_{k,k}^{-{\beta}_{k,k}}}{m}\!\right)^{\!m}}\!\!\right) \nonumber\\ & ~~~~~~~~~~~~~~~~~~~~~~\cdot\prod_{\ell\neq k}^K \frac{1}{\left(1+{z d_{k,\ell}^{-{\beta}_{k,\ell}}}\right)^{a_{j,\ell}}}   {\rm d}z. \label{theorem1}
\end{align}
\begin{proof}
    See Appendix A.
\end{proof}




\end{theorem}

Theorem 1 provides a formula for computing the ergodic sum-spectral efficiency in terms of several network parameters, including the SNR, the Nakagami fading parameter $m$, the long-term path loss of all links, denoted by $\left\{d_{k,\ell}^{-{\beta}_{k,\ell}}\right\}$ for $k,\ell\in \{1,\ldots,K\}$, and the scheduling action ${\bf a}_j\in \mathcal{A}$. For a given set of network parameters, we can identify the optimal scheduling action that maximizes the ergodic sum-spectral efficiency by solving the combinatorial optimization problem in \eqref{eq:ESE_opt}. However, the computational complexity of this optimization problem increases exponentially with the number of D2D links $K$, which makes it unsuitable for dense D2D networks. Therefore, we need to develop a low-complexity algorithm to solve this problem, which will be presented in the following subsection.

\subsection{Low-Complexity Link Scheduling Algorithm} \label{sec:low_comp}

In this subsection, we propose the low-complexity algorithm called L-QuasiOpt that finds quasi-optimal cooperative link scheduling action. The L-QuasiOpt algorithm finds the quasi-optimal cooperative link scheduling action within the short trials than $2^K$ by clustering the D2D links with path-loss knowledge and identifying the link scheduling action for each cluster. In each cluster, the link scheduling action is chosen to maximize the cluster utility function which is the ergodic sum-spectral efficiency averaged over inter-cluster interference distributions.

For the clustering strategy, we consider the well-known hierarchical clustering algorithm \cite{Hier2016} which clusters the links by using the pairwise distances between the links. Specifically, hierarchical clustering constructs a hierarchy of clusters by merging them iteratively based on their pairwise distances, generating a dendrogram that illustrates cluster relationships. Then, D2D links are assigned to clusters using the dendrogram structure. Since the BS has knowledge of path-loss $d_{k,\ell}^{-\frac{\beta_{k,\ell}}{2}}$, clustering techniques that require exact locations 
of D2D links, such as K-means clustering algorithm, is not applicable. In traditional interference management methods \cite{Jafar2014,Naderializadeh2015,Doumiati2019}, the sub-networks are constructed to minimize interference within each sub-network and use different frequency resources for each sub-networks to manage the inter-sub-networks interference. However, these approaches are not applicable when all D2D links share the same frequency and time resources. To address this challenge, we cluster the severely interfering D2D links into the same cluster, which increases intra-cluster interference but reduces inter-cluster interference. We then schedule the links within each cluster to manage the intra-cluster interference.

\begin{figure*}[!t]
\begin{center}
\includegraphics[width=13cm]{ 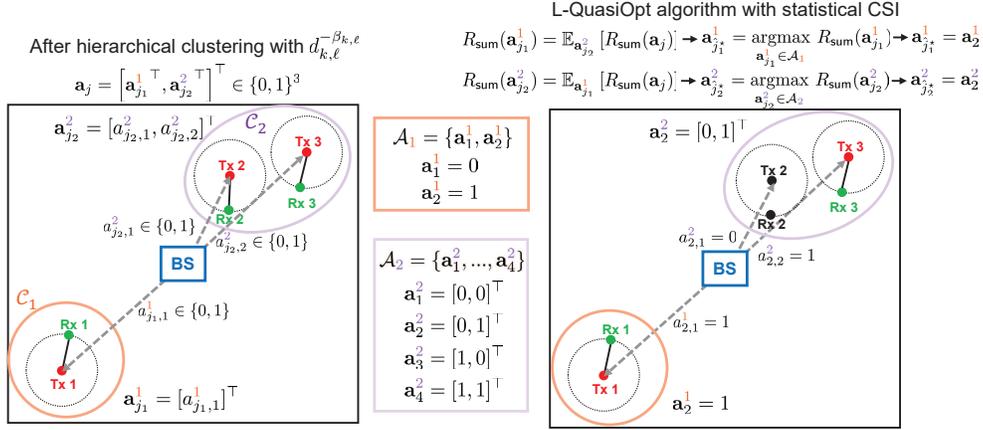}
\end{center}
\caption{An illustration of the L-QuasiOpt algorithm when $K=3$ and $C=2$.}\label{fig:L-QuasiOpt}
\end{figure*}

The BS first clusters the $K$ links into $C$ clusters $\mathcal{C}_1,\mathcal{C}_2,\ldots,\mathcal{C}_C$, by applying hierarchical clustering algorithm with path-loss knowledge $d_{k,\ell}^{-\frac{\beta_{k,\ell}}{2}}$.
Let $K_c$ be the number of D2D links in the cluster $\mathcal{C}_c$, where $\sum_{c=1}^C K_c = K$, and let the set of users in the cluster $\mathcal{C}_c$ be $\mathcal{K}_c$. Then, we define a cluster scheduling action vector of cluster $\mathcal{C}_c$ as $\mathbf{a}_{j_c}^{c}=\left[a_{j_c,1}^c,\ldots,a_{j_c,K_c}^{c} \right]^{\top} \in \{0,1\}^{K_c}$ where $j_c \in \{1,\ldots,2^{K_c}\}$ is a scheduling action index of cluster $\mathcal{C}_c$. By collecting all possible scheduling actions of cluster $\mathcal{C}_c$, we define the scheduling action space of cluster $\mathcal{C}_c$ as
\begin{align}
    \mathcal{A}_c = \bigl\{ \mathbf{a}_{1}^c,\mathbf{a}_{2}^c,\ldots, \mathbf{a}_{2^{K_c}}^c  \bigr\}.
\end{align}
Further, by concatenating the cluster scheduling action $\mathbf{a}_{j_c}^{c}$ for all $c\in \{1,2,\ldots,C\}$, the cooperative scheduling action is expressed as
\begin{align}
    \mathbf{a}_{j} = \left[{\mathbf{a}_{j_{1}}^{1}}^{\top},{\mathbf{a}_{j_{2}}^{2}}^{\top},\ldots,{\mathbf{a}_{j_C}^{C}}^{\top}\right]^{\top}\in\{0,1\}^K.
\end{align}
Then, we define the cluster utility function of cluster $\mathcal{C}_c$ for the cluster scheduling action $\mathbf{a}_{j_c}^{c}$ as

\begin{align}
R_{\sf sum} (\mathbf{a}_{j_c}^{c}) = \mathbb{E}_{\mathbf{a}_{j_{s}}^{s}} \left[ R_{\sf sum} ({\bf a}_j )\right],
\end{align}
which is the ergodic sum-spectral efficiency averaged over interference distributions of other clusters $\mathcal{C}_s$ for $s \ne c$. Since obtaining the interference distribution of other clusters is unrealistic, we assume that the actions of other clusters are Bernoulli distributed as Bernoulli$(0.5)$, i.e., $\mathbb{E}[a_{j_{s},i}^s]=0.5$ for $s \ne c$ and $i \in \{1,\ldots,K_{s}\}$. {Then, the cluster utility function $R_{\sf sum} (\mathbf{a}_{j_c}^{c})$ is computed as
\begin{align}
&R_{\sf sum} (\mathbf{a}_{j_c}^c)  = \mathbb{E}_{\mathbf{a}_{j_s}^s} \left[ R_{\sf sum} ({\bf a}_j )\right] \nonumber \\ &= \mathbb{E}_{a_{j,k},a_{j,\ell}} \Bigg[ \log_2{e}\sum_{k=1}^K  a_{j,k}  \int_{0}^{\infty}\!\!\frac{e^{-\frac{z}{{\sf snr}}}}{z} \! \nonumber\\ &~~~~~~\cdot   \left(\!\!1-\!{\left(1+\frac{zd_{k,k}^{-{\beta}_{k,k}}}{m}\!\right)^{-m}}\!\right) \prod_{\ell\neq k}^K \frac{1}{1+{z d_{k,\ell}^{-{\beta}_{k,\ell}}}a_{j,\ell}}   {\rm d}z \Bigg] \nonumber \\
&\overset{(a)}{=} \log_2{e}\sum_{k=1}^K \mathbb{E}_{a_{j,k}}\left[a_{j,k}\right] \int_{0}^{\infty}\!\!\frac{e^{-\frac{z}{{\sf snr}}}}{z} \!\left(\!\!1-\!{\left(1+\frac{zd_{k,k}^{-{\beta}_{k,k}}}{m}\!\right)^{-m}}\!\right)  \nonumber \\&~~~~~~~~~~~~~~~~~~~~~~~~~\prod_{\ell\neq k}^K \mathbb{E}_{a_{j,\ell}} \left[ \frac{1}{1+{z d_{k,\ell}^{-{\beta}_{k,\ell}}}a_{j,\ell}} \right]  {\rm d}z, \label{eq:cluster_uf}
\end{align}
where (a) follows from the independence of $a_{j,k}$ and $a_{j,\ell}$ for $\ell \ne k$. Since the distributions of actions of other clusters $\mathcal{C}_s$ for $s\ne c$ are assumed to be Bernouilli(0.5), the expectations in \eqref{eq:cluster_uf} are obtained by
\begin{align}    \mathbb{E}_{a_{j,k}}\left[a_{j,k}\right] = 
    \begin{cases}
        a_{j,k} & \text{if}~ k\in \mathcal{K}_c \\
        0.5 & \text{if}~ k \notin \mathcal{K}_c,
    \end{cases}        
    \end{align}
and
\begin{align}
    \mathbb{E}_{a_{j,\ell}} \left[ \frac{1}{1+{z d_{k,\ell}^{-{\beta}_{k,\ell}}}a_{j,\ell} } \right] =
    \begin{cases}
        \frac{1}{1+{z d_{k,\ell}^{-{\beta}_{k,\ell}}}a_{j,\ell}} & \text{if}~\ell \in \mathcal{K}_c \\ 
        \frac{2+zd_{k,\ell}^{-\beta_{j,\ell}}}{2+2zd_{k,\ell}^{-\beta_{j,\ell}}} &\text{if}~\ell \notin \mathcal{K}_c
    \end{cases}.
\end{align}}

\begin{algorithm} [t]
	\caption{L-QuasiOpt algorithm.}\label{alg:lowcomp}
	{\small{\begin{algorithmic}[1]
    \STATE Link clustering using hierarchical clustering algorithm with $d_{k,\ell}^{-\beta_{k,\ell}}$.
    \FOR{$c \in \{1,\ldots,C\}$}
        \FOR{$j_c \in \{1,\ldots,2^{K_c}\}$}
            \STATE Compute the cluster utility function $R_{\sf sum} (\mathbf{a}_{j_c}^c)$: $R_{\sf sum} (\mathbf{a}_{j_c}^c) = \mathbb{E}_{\mathbf{a}_{j_s}^s} \left[ R_{\sf sum} ({\bf a}_j )\right]$.
        \ENDFOR
        \STATE Select cluster link scheduling action $\mathbf{a}_{\hat{j}_c^\star}^c$ that maximizes the cluster utility function: \\ $\mathbf{a}_{\hat{j}_c^\star}^c = \underset{\mathbf{a}_{j_c}^c \in \mathcal{A}_c}{\mathrm{argmax}} ~ R_{\sf sum} (\mathbf{a}_{j_c}^c)$
    \ENDFOR
    \STATE Obtain quasi-optimal cooperative scheduling action $\mathbf{a}_{\hat{j}^\star}$ by concatenating $\mathbf{a}_{\hat{j}_c^\star}^c$: \\ $\mathbf{a}_{\hat{j}^\star} = \left[{\mathbf{a}_{\hat{j}_{1}^\star}^{1}}^{\top};{\mathbf{a}_{\hat{j}_{2}^\star}^{2}}^{\top};\ldots;{\mathbf{a}_{\hat{j}_{C}^\star}^{C}}^{\top}\right]^{\top}\in\{0,1\}^K$.
	\end{algorithmic}}}
\end{algorithm}


The BS selects the cluster link scheduling action that maximizes the cluster utility function as
\begin{align}
    \mathbf{a}_{\hat{j}_{c}^\star}^{c} = \underset{\mathbf{a}_{j_{c}}^{c} \in \mathcal{A}_{c}}{\mathrm{argmax}} ~ R_{\sf sum} (\mathbf{a}_{j_{c}}^{c}). \label{eq:low_opt}
\end{align}

Then, by concatenating the cluster link scheduling actions $\mathbf{a}_{\hat{j}_{c}^\star}^{c}$ for all $c\in \{1,2,\ldots,C\}$, the quasi-optimal cooperative scheduling action is obtained as
\begin{align}
    \mathbf{a}_{\hat{j}^\star} = \left[{\mathbf{a}_{\hat{j}_{1}^\star}^{1}}^{\top},{\mathbf{a}_{\hat{j}_{2}^\star}^{2}}^{\top},\ldots,{\mathbf{a}_{\hat{j}_{C}^\star}^{C}}^{\top}\right]^{\top}\in\{0,1\}^K.
\end{align}
The proposed L-QuasiOpt algorithm is summarized in Algorithm \ref{alg:lowcomp}.
The L-QuasiOpt algorithm reduces the $2^K$ computational complexity of optimization problem in \eqref{eq:ESE_opt} to $\sum_{c=1}^{C} 2^{K_c}$ by solving the optimization problem in \eqref{eq:low_opt} for each cluster.


 \begin{figure*}[!t]
\begin{center}
\includegraphics[width=12cm]{ 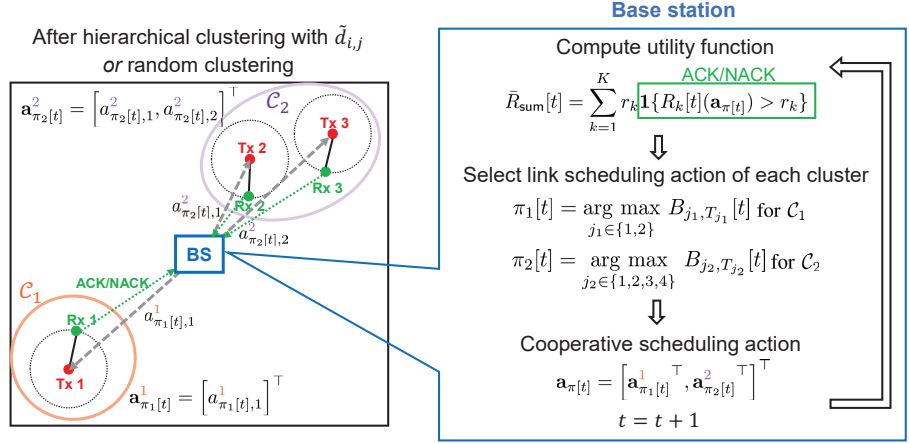}
\end{center}
\caption{An illustration of the BanditLinQ algorithm when $K=3$ and $C=2$.}\label{fig:Clustered_UCB}
\end{figure*}

\section{Quasi-optimal Link Scheduling with One-bit Feedback} 

In this section, we propose a D2D link scheduling algorithm that uses one-bit feedback information per D2D receiver. The proposed algorithm does not require any knowledge of network topology and channel fading distributions unlike previous D2D link scheduling algorithms in \cite{Huynh2016,Lee2015,Yin2016,Maghsudi2015,Yu2011,Jiang2018,Naqvi2018,Takshi2018,Khuntia2019,Gao2019,Navid2014}. Specifically, in each fading block, the proposed link scheduling strategy selects the cooperative action index $\pi[t] \in \{1,\ldots,J\}$ from the collection of the utility functions, i.e., $\{{\bar R}_{\sf sum}[t]({\bf a}_{\pi[t]})\}_{u=1}^{t-1} \rightarrow \pi[t] \in \{1,\ldots,J\}$, which are obtained with one-bit feedback information.

To identify the optimal cooperative scheduling action that maximizes the average sum-throughput, the BS has to obtain a mean utility function $\mathbb{E}\left[{\bar R}_{\sf sum}[t]({\bf a}_{j}) \right]$ for each cooperative link scheduling action only with the collection of ACK/NACK feedback. 
{We learn the mean utility function without knowing fading distribution and network topology in advance by using MAB framework. The reason for employing MAB framework is rooted in its capability to effectively handle scenarios where the underlying distribution of parameters, such as the fading distribution and network topology, is uncertain or unknown for interference management. MAB algorithms strike a balance between exploring and exploitation, providing theoretical guarantees of performance. Therefore, applying the MAB framework is an effective solution in scenarios where we need to find scheduling action that maximizes ergodic sum-throughput without complete knowledge.}
As one of the optimal exploration strategies in MAB framework, we consider the UCB1 algorithm \cite{Auer2002}. In each fading block, the UCB1 algorithm evaluates each action's potential to maximize the utility function based on previous observations. The action with the highest potential for maximizing the utility function is then selected. With this strategy, the UCB1 algorithm finds a sequence of the optimal cooperative scheduling actions that maximize the mean utility function. In other words, the UCB1 algorithm finds the optimal cooperative scheduling action sequence that minimizes the cumulative regret which is defined as
\begin{align}
	{\sf Reg}_{j^\star}[n] = n\mu_{j^{\star}} - \sum_{t=1}^n\mathbb{E}\left[ {\bar R}_{\sf sum}[t]({\bf a}_{\pi[t]}) \right], \label{eq:cumregret}
\end{align}
where $\mu_{j^\star}=\mathbb{E}\left[\Bar{R}_{\sf sum}[t](\mathbf{a}_{j^\star})\right]$ is the maximum mean utility function. 
The cumulative regret in \eqref{eq:cumregret} represents the cumulative differences between the average sum-throughput given the optimal cooperative scheduling action and the average sum-throughput of the selected cooperative scheduling actions. 

Although the UCB1 algorithm produces optimal logarithmic regret, it suffers from a significant limitation whereby the number of necessary explorations grows exponentially with the number of D2D links $K$. To handle this challenge, we propose the BanditLinQ link scheduling algorithm as a scalable D2D link scheduling algorithm for large $K$. The BanditLinQ algorithm clusters D2D links into multiple D2D clusters with one-bit feedback information and selects link scheduling action per cluster that maximizes the cluster utility function.

\begin{algorithm} [t]
	\caption{Link clustering algorithm.}\label{alg:clustering}
	{\small{\begin{algorithmic}[1]
    \IF{ Clustering with one-bit feedback}
    	\FOR {$t=1,\ldots, T_{\sf clust}$}
            \FOR{$i,j \in \{1,\ldots,K\}$}
                \STATE The D2D receiver estimates their received SNR and INR.
                \STATE The $i$th D2D receiver sends one-bit information to the BS:
                    $\mathbf{1}\{\text{INR}_{ij}[t]< \text{SNR}_i^\eta[t] , \text{INR}_{ji}[t] < \text{SNR}_i^\eta[t] \}$     
                \STATE Obtain one-bit information of distance $p_{i,j}[t]$:  $p_{i,j}[t] = \mathbf{1}\{\text{INR}_{ij}[t]< \text{SNR}_i^\eta[t] , \text{INR}_{ji}[t] < \text{SNR}_i^\eta[t] \}$.
            \ENDFOR
        \ENDFOR
        \STATE Hierarchical clustering with $\tilde{p}_{i,j} = \frac{1}{T_{\sf clust}}\sum_{t=1}^{T_{\sf clust}} p_{i,j}[t]$.

    \ELSIF{Random clustering}
        \STATE BS randomly divides $K$ links into $C$ clusters. 
    \ENDIF
	\end{algorithmic}}}
\end{algorithm}

\subsection{D2D Link Clustering Strategies} \label{sec:clustering}
\subsubsection{D2D Link Clustering with One-Bit Feedback}
We first describe the link clustering strategy that uses one-bit feedback per each D2D receiver. {We use the hierarchical clustering algorithm to cluster the severely interfering D2D links into the same cluster. Since the hierarchical clustering algorithm clusters links with small pairwise distances into the same cluster, the pairwise distance in hierarchical clustering has to be set small when the pairwise interference is large.
To acquire approximations of the pairwise interference levels using one-bit feedback information, we assume that all D2D links use different frequency bands each other and their transmissions cause no interference to other links in the clustering phase $T_{\sf clust}$. Then the D2D receiver estimates their received SNR and interference-to-noise ratio (INR) during the clustering phase $T_{\sf clust}$.
The one-bit information of interference level between the $i$th D2D link and the $j$th D2D link, denoted as $p_{i,j}$, is obtained inspired by the ITLinQ algorithm \cite{Navid2014}. Specifically, the $i$th D2D receiver sends one-bit information to the BS based on the following two conditions:
\begin{align}
    p_{i,j}=\mathbf{1}\{\text{INR}_{ij}[t]< \text{SNR}_i^\eta[t] , \text{INR}_{ji}[t] < \text{SNR}_i^\eta[t] \},
\end{align}
where $\text{INR}_{ij}$ denotes the INR of $j$th transmitter at the $i$th receiver for $i,j\in \{1,\ldots,K\}$ and $\eta$ is a hyper-parameter. If $\text{INR}_{ij}[t]\ge \text{SNR}_i^\eta[t]$ or $\text{INR}_{ji}[t] \ge \text{SNR}_i^\eta[t]$, the $i$th link causes or receives much interference to and from the $j$th link compared to the SNR. Therefore, the $i$th receiver sends pairwise distance information as $p_{i,j}[t]=0$ to cluster these two links into the same cluster. Conversely, if $\text{INR}_{ij}[t]< \text{SNR}_i^\eta[t]$ and $\text{INR}_{ji}[t] < \text{SNR}_i^\eta[t]$, the interference between two links is small compared to the SNR. In this scenario, the $i$th receiver sends pairwise distance information as $p_{i,j}[t]=1$ not to cluster these two links into the same cluster. After repeating this process during the clustering phase $T_{\sf clust}$, the BS computes an average of the one-bit pairwise distance information as $\tilde{p}_{i,j} = \frac{1}{T_{\sf clust}}\sum_{t=1}^{T_{\sf clust}} p_{i,j}[t]$. Finally, $\tilde{p}_{i,j}$ is used as pairwise distance information for the hierarchical clustering algorithm.}

\subsubsection{D2D Link Clustering without Any Knowledge}
In practice, however, it is unrealistic to acquire accurate SNR and INR. Therefore, we also consider a random clustering method that simply divides links into multiple small clusters. Since the random clustering method does not require feedback information from the D2D receiver during the clustering phase, $T_{\sf clust}=0$. The proposed link clustering algorithms are summarized in Algorithm \ref{alg:clustering}. In the following subsection, we propose the BanditLinQ link scheduling algorithm that finds quasi-optimal cooperative link scheduling action with clustered links.

\subsection{BanditLinQ Algorithm for D2D Link Scheduling}
Once the $K$ links are clustered into multiple clusters $\mathcal{C}_1,\mathcal{C}_2,\ldots,\mathcal{C}_C$ the BS selects a scheduling action in each fading block applying BanditLinQ algorithm. The key idea of the BanditLinQ algorithm is to apply the UCB1 algorithm per cluster with an empirical cluster utility function which is the empirically averaged instantaneous sum-throughput. 

Let ${\pi_{c}[u]} \in \{1,\ldots,2^{K_c}\}$ be a scheduling action index of cluster $\mathcal{C}_c$ at fading block $u$. By concatenating the cluster link scheduling actions $\mathbf{a}_{\pi_{c}[u]}^{c}$ for all $c\in \{1,2,\ldots,C\}$, the cooperative scheduling action at fading block $u$ is defined as
\begin{align}
    \mathbf{a}_{\pi[u]} = \left[{\mathbf{a}_{\pi_1[u]}^{1}}^{\top},{\mathbf{a}_{\pi_2[u]}^{2}}^{\top},\ldots,{\mathbf{a}_{\pi_C[u]}^{C}}^{\top}\right]^{\top}\in\{0,1\}^K.
\end{align}
We then define $T_{j_c}[t]$ as the number of trials for the cluster link scheduling action $\mathbf{a}_{j_c}^{c}$ before the $t$th fading blocks: 
\begin{align}
    T_{j_c}[t] = \sum_{u=1}^{t-1}\mathbf{1}{\{\pi_c[u] = j_c\}}.
\end{align}
To obtain the estimates of the cluster utility functions, the BS first computes the utility function ${\bar R}_{\sf sum}[u]({\bf a}_{\pi[u]})$ with ACK/NACK feedback information as in \eqref{eq:one_bit_uility}. Then, the empirical cluster utility function $\hat{\mu}_{j_c,T_{j_c}}[t]$ for the cluster scheduling action $\mathbf{a}_{j_c}^c$ at fading block $t$ is obtained by taking the sample average over the utility function ${\bar R}_{\sf sum}[u]({\bf a}_{\pi[u]})$ only when the cluster scheduling action $\mathbf{a}_{j_c}^c$ is selected up to the $t-1$th fading block as follows:
\begin{align}
    \hat{\mu}_{j_c,T_{j_c}}[t] = \sum_{u=1}^{t-1} \frac{{\bar R}_{\sf sum}[u]({\bf a}_{\pi[u]})\mathbf{1}{\{\pi_c[u] = j_c\}}}{T_{j_c}[t]}.
\end{align}
This empirical cluster utility function is the sum-throughput for the cluster scheduling action $\mathbf{a}_{j_c}^c$ of cluster $\mathcal{C}_c$, which is empirically averaged over interferences from other clusters.

Applying Hoeffding’s inequality and Chernoff bounds, the BanditLinQ algorithm harnesses the upper bound of the confidence interval to select the scheduling action at each fading block. The upper confidence bound of action $j_c$ of cluster $\mathcal{C}_c$ is defined as
\begin{align}
    B_{{j_c},T_{j_c}}[t] =  \hat{\mu}_{{j_c},T_{j_c}}[t] + \sqrt{\frac{\alpha  2^{K-K_c}\left(\sum_{k=1}^K r_k\right)^2 \ln t}{2 T_{j_c}[t]}}, \label{eq:ucbterm}
\end{align}
for some $\alpha>0$.
The BS selects the scheduling action of each cluster that provides the maximum value of $B_{{j_c},T_{j_c}}[t]$ at each fading block $t$:
\begin{align}
    \pi_c[t] = \underset{j_c \in \{1,\ldots,2^{K_c}\}}{\mathrm{arg~max}} ~B_{{j_c},T_{j_c}}[t].
\end{align}

The proposed algorithm is summarized in Algorithm \ref{alg:Cucb}. 
Taking $K=3$ as an example, Fig. \ref{fig:Clustered_UCB} shows the process of the BanditLinQ algorithm when $C=2$. The BS first clusters the $3$ D2D links into two clusters $\mathcal{C}_1$ and $\mathcal{C}_2$ then computes the empirical cluster utility function of each cluster using the ACK/NACK feedback information from 10 D2D receivers. Subsequently, the BS selects the link scheduling action $\mathbf{a}_{\pi_c[t]}^c$ of each cluster that maximizes the upper confidence bound.

\begin{algorithm} [t]
	\caption{BanditLinQ link scheduling algorithm.}\label{alg:Cucb}
	{\small{\begin{algorithmic}[1]
	\FOR {$t=1,\ldots, T$}
	\IF {$t<T_{\sf clust}$}
	\STATE Link clustering with Algorithm 2
	\ELSE
	    \FOR{$c = 1,\ldots, C$}
	        \IF {$ T_{j_c} [t] =0$ for $j_c=\{1,\ldots,2^{K_c}\}$}
            \STATE $\pi_c[t] = j_c$. ~~~~ (\text{initialization})
            \ELSE
            \STATE Compute the number of trials of the scheduling action $j_c$ until fading block $t$: \nonumber \\ $T_{j_c} [t] = \sum_{u=1}^{t-1}\mathbf{1}{\{\pi_c[u] = j_c\}}$ for $j_c=\{1,\ldots,2^{K_c}\}$.
            \STATE Compute the cluster utility function for the scheduling action $j_c$ at fading block $t$: \nonumber \\ $\hat{\mu}_{j_c,T_{j_c}}[t] = \sum_{u=1}^{t-1} \frac{{\bar R}_{\sf sum}[u]({\bf a}_{\pi[u]})\mathbf{1}{\{\pi_c[u] = j_c\}}}{ T_{j_c}[t]}$.
            \STATE Select the scheduling action in each cluster that provides maximum upper bound in the $t$th fading block: $\pi_c[t] = \underset{j_c \in \{1,\ldots,2^{K_c}\}}{\mathrm{arg~max}}~ B_{{j_c},T_{j_c}}[t]$.
            \ENDIF
	    \ENDFOR
	   \STATE Compute the cooperative scheduling action by gathering $\pi_c[t]$ for all $c = \{1,\ldots,C\}$: \nonumber \\ $ \mathbf{a}_{\pi[u]} = \left[{\mathbf{a}_{\pi_1[u]}^{1}}^{\top},{\mathbf{a}_{\pi_2[u]}^{2}}^{\top},\ldots,{\mathbf{a}_{\pi_C[u]}^{C}}^{\top}\right]^{\top}\in\{0,1\}^K$.
        \STATE Compute utility function ${\bar R}_{\sf sum}[t]({\bf a}_{\pi [t]})$:  ${\bar R}_{\sf sum}[t]({\bf a}_{\pi [t]}) = \sum_{k=1}^{K} r_k \mathbf{1}\{R_{k}[t]({\bf a}_{\pi [t]})>r_k\}$.
	\ENDIF
	\ENDFOR
	\end{algorithmic}}}
\end{algorithm}

\subsection{Optimality of the BanditLinQ Algorithm}
After a sufficient number of fading blocks, the empirical cluster utility function $\hat{\mu}_{j_c,T_{j_c}}[t]$ converges to the true cluster utility function $\mu_{j_c}$, which is the ergodic sum-throughput averaged over inter-cluster interference distributions:
\begin{align}
    \mu_{j_c} = \mathbb{E}_{\mathbf{a}_{j_\ell}^\ell}\left[\mathbb{E}_{|h_{k,\ell}[t]|^2}\left[{\bar R}_{\sf sum}[t]({\bf a}_{j})\right]\right],
\end{align}
where $\ell \ne c$.
We then define $\mu_{\hat{j}_c^\star}$ as a ergodic sum-throughput for the quasi-optimal cluster link scheduling action as
\begin{align}
    \mu_{\hat{j}_c^\star}= \underset{j_c\in \{1,\ldots,2^{K_c}\}}{\mathrm{max}}~\mathbb{E}_{\mathbf{a}_{j_\ell}^\ell}\left[\mathbb{E}_{|h_{k,\ell}[t]|^2}\left[{\bar R}_{\sf sum}[t]({\bf a}_{j})\right]\right]. 
\end{align}
Whereas, the ergodic sum-throughput for the optimal cluster link scheduling action is defined as
\begin{align}
    \mu_{j_c^\star}=\underset{j_c\in \{1,\ldots,2^{K_c}\}}{\mathrm{max}}\mathbb{E}_{|h_{k,\ell}[t]|^2}\left[{\bar R}_{\sf sum}[t]({\bf a}_{j})\big| j_\ell=j_\ell^\star,~\ell \ne c\right],
\end{align}
which is the maximum ergodic sum-throughput given the optimal actions of other clusters.

The BanditLinQ algorithm finds quasi-optimal cooperative action $\mathbf{a}_{\hat{j}^\star}$ with logarithmic regret, where the obtained quasi-optimal cooperative action has a constant sum-throughput gap from optimal cooperative action and it is proved by Theorem 2.

\begin{theorem}
 We define $\delta_c=\mu_{j_c^\star} - \mu_{\hat{j}_c^\star}$ as the gap between the cluster utility function of optimal cluster link scheduling action $\mathbf{a}_{j_c^{\star}}$ and quasi-optimal cluster link scheduling action $\mathbf{a}_{\hat{j}_c^{\star}}$. Further, we define $\Delta_{j_c}= \mu_{\hat{j}_c^{\star}}-\mu_{j_c}$ as the gap between the cluster utility function of quasi-optimal action $\mathbf{a}_{\hat{j}_c^\star}$ and action $\mathbf{a}_{j_c}$ of the cluster $\mathcal{C}_c$. $\Delta^{\sf max}= \underset{j\in [2^K]}{\mathrm{\max}}\mu_{\hat{j}^{\star}}-\mu_{j}$ denotes maximum gap between the utility function of quasi-optimal action $\mathbf{a}_{\hat{j}^\star}$ and action $\mathbf{a}_{j}$.
 
Then, the BanditLinQ algorithm finds quasi-optimal cooperative action $\mathbf{a}_{\hat{j}^\star}$ with logarithmic regret as
\begin{align}
    &\mathbb{E}\left[{\sf Reg}_{\hat{j}^\star}[n]\right] \nonumber\\ &\le \Delta^{\sf max} \sum_{c=1}^C \sum_{\substack{j_c=1, \\ j_c\ne \hat{j}_c^\star}}^{2^{K_c}} \left(\frac{2^{K-K_c+3}(\sum_{k=1}^K r_k)^2  \ln n}{\Delta_{j_c}^2}  + 1 + \frac{\pi^2}{3}\right),
\end{align}
where a maximum ergodic sum-throughput gap between the quasi-optimal cooperative link scheduling action $\mathbf{a}_{\hat{j}^\star}$ and optimal cooperative link scheduling action $\mathbf{a}_{j^\star}$ is $\frac{1}{C}\sum_{c=1}^{C} \delta_c$, i.e.,
\begin{align}
    \mu_{j^\star} - \mu_{\hat{j}^\star} \le \frac{1}{C}\sum_{c=1}^{C} \delta_c.
\end{align}
\end{theorem}

\begin{proof}
    See Appendix B.
\end{proof}

The main difference between the regret of the UCB1 algorithm \cite{Auer2002} and that of the BanditLinQ algorithm lies in the uncertainty term. The BanditLinQ algorithm requires more exploration to average out inter-cluster interferences, resulting in a larger uncertainty term than that of the UCB1 algorithm. The uncertainty term in \eqref{eq:ucbterm} of the BanditLinQ algorithm is derived under the assumption that the throughput of one link is affected by all other links. If the interferences between the clusters are very small as a result of clustering, the exploration cost, i.e., $2^{K-K_c}$, in the uncertainty term, can be greatly reduced and the regret can also be reduced. 

Further, there is a trade-off between the initialization cost and the exploration cost. Specifically, for small $C$, the number of D2D links per cluster ${K_c}$ is large. Then the exploration cost $2^{K-K_c}$ in \eqref{eq:ucbterm} becomes small but the initialization cost is large since the initialization requires $ 2^{{\max} K_c}$ fading blocks. In TABLE I, we show the example of the initialization cost and the exploration cost by cluster size when $K=50$, assuming that each cluster has the same number of D2D links. The cluster size should be determined to ensure that both the exploration cost and initialization cost remain reasonable by taking into account the number of fading blocks $T$.

\setlength{\tabcolsep}{3pt}
\begin{table}[!t]
\caption{Comparison of exploration cost \\and initialization cost by cluster size.\label{Table1}} 
\centering
\begin{tabular}{c|c|c}
\hline
 ~& Exploration cost &Initialization cost \\ \hline \hline
 No clustering ($C=1, K_c=50$)& 1 & $2^{50}$ fading blocks \\ \hline
 Clustering ($C=5, K_c=10$)& $2^{40}$ & $2^{10}$ fading blocks\\ \hline
  Clustering ($C=10, K_c=5$)& $2^{45}$ & $2^5$ fading blocks\\ \hline
\end{tabular}
\end{table}

{\bf Remark 2 (Scheduling complexity reduction):} The proposed BanditLinQ algorithm does not require any knowledge of the fading distribution and network topology. Instead, the empirical cluster utility function is obtained with repeated one-bit feedback per D2D link by MAB framework. Although the UCB1 algorithm finds the optimal link scheduling action with only one-bit feedback per D2D link, the number of actions that have to be explored exponentially increases with $K$. Specifically, the initialization phase for operating UCB1 algorithm requires $2^K$ fading blocks, which highly degrades the regret and sum-throughput performances. The proposed BanditLinQ algorithm reduces action space from $2^K$ to $\sum_{c=1}^C 2^{K_c}$ by clustering the entire links into small clusters. Unlike a classical CMAB framework \cite{Ortiz2019,hakami2022,Ontan2013} that breaks the whole MAB problem into several smaller MAB problems and updates the cluster utility function of each small MAB problem with its own cluster's throughput, the BanditLinQ algorithm updates the cluster utility function of each cluster with the sum-throughput. Therefore, the BanditLinQ algorithm finds quasi-optimal cooperative scheduling action in a centralized manner without the computational complexity of selecting the super-arm.

\section{Simulation Results}


In this section, we compare the performance of the proposed algorithms with the existing link scheduling algorithms in terms of the ergodic sum-throughput and average sum-spectral efficiency. In our simulation,
the path-loss exponent of each link is drawn from an IID uniform random variable, i.e., $\beta_{k,\ell} \sim \mathcal{U}(3.5,4.5)$. We drop $K$ D2D transmitters in the $\mathbb{R}^2$ plane and the positions are fixed over the fading blocks $T=5000$. The D2D receiver is dropped at a fixed distance $d_{k,k}=50 \text{m}$ away from the corresponding D2D transmitter in the initial fading block. The transmit power is set to $P = 0.08 \text{mW}$, and the noise power spectral density is considered to be $-143.97 \text{dBm}$.


We compare the BanditLinQ algorithm with the following link scheduling strategies: 
\begin{itemize}
	\item Information-theoretic link scheduling (ITLinQ) \cite{Navid2014}: each D2D link is scheduled if it does not cause and receive much interference to and from the other links. Specifically, if the following two conditions are satisfied, the $j$th link is activated:
	\begin{align}
	    \text{INR}_{ji} \le \text{SNR}_j^\eta ~~\forall{i<j}, ~
	    \text{INR}_{ij} \le \text{SNR}_j^\eta ~~\forall{i<j}.
	\end{align}
	\item Distributed on-off power control algorithm (D-OnOff) \cite{Lee2015}: the $k$th links are scheduled if 
	\begin{align}
	    	|h_{k,k}^2 d_{k,k}^{-\beta_{k,k}}|^2>\frac{-\ln \left(\min \left\{\frac{\text{sinc} \left(\frac{2}{\beta_{k,k}}\right)}{\pi \lambda \kappa ^{\frac{2}{\beta_{k,k}}}d_{k,k}^2},1\right\}\right)}{d_{k,k}^{\beta_{k,k}}}, \label{eq:D-onoff}
	\end{align}
	where $\lambda$ denotes the density of devices in network and $\kappa$ denotes target signal-to-interference ratio (SIR).
    \item D2D-CMAB \cite{Ortiz2019}: action for each link is selected in a distributed manner by using the NS strategy. The super-MAB, which selects the cooperative action, follows the greedy policy. Local-MAB, which is the distributed MAB per link, is used to explore new cooperative action and follows the preference-based policy.
	\item Random scheduling: the BS selects a cooperative scheduling action randomly out of $2^K$ actions at every fading block $t$.
	\item No scheduling: the BS schedules all the $K$ D2D links at every fading block $t$. 
\end{itemize}

\begin{figure}[t]
\begin{center}
\includegraphics[width=8.5cm]{ 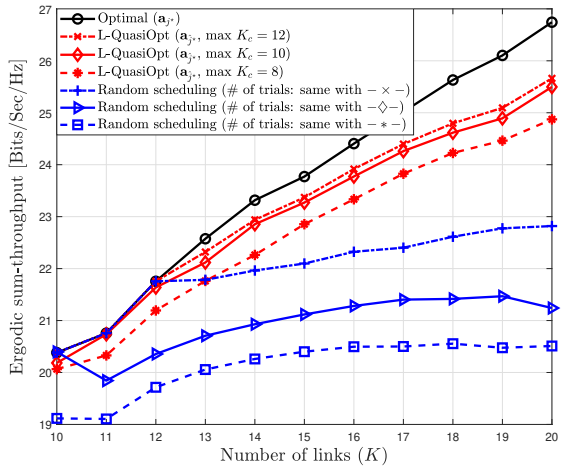}
\end{center}
\caption{{Ergodic sum-throughput $\mathbb{E}\left[\Bar{R}_{\sf sum}[t](\mathbf{a}_j)\right]$ comparison for $m=1$.}}\label{fig:ergodicST}
\end{figure}

Fig. \ref{fig:ergodicST} shows the comparison of ergodic sum-throughput $\mathbb{E}\left[\Bar{R}_{\sf sum}[t](\mathbf{a}_j)\right]$ with the selected cooperative scheduling action of each algorithm. We assume that the BS has knowledge of statistical CSI and network parameters for computing the ergodic sum-throughput. In this simulation, we drop $K$ D2D links uniformly in a  0.5 km $\times$ 0.5 km $\mathbb{R}^2$ plane for a more interfering environment with a small number of links. We set the target rate as $r_k = 5$ for all D2D links and set Nakagami parameter as $m=1$.  In the case of Optimal, $\mathbf{a}_{j^\star}$ is obtained by solving the combinatorial optimization problem with $2^K$ trials. It becomes the upper bound of ergodic sum-throughput performance. The L-QuasiOpt algorithm in section \ref{sec:low_comp} obtains quasi-optimal link scheduling action $\mathbf{a}_{\hat{j}^\star}$ after $\sum_{c=1}^C 2^{K_c}$ trials. We compute the cluster utility function of the L-QuasiOpt algorithm by assuming that the actions of other clusters follow Bernoulli(0.5) distribution. { We also present the ergodic sum-throughput averaged over the number of D2D links $K$ from $K=1$ to $K=20$ in TABLE \ref{Table:EST}. The number of trials for the L-QuasiOpt algorithm in our simulation by the maximum cluster size is summarized in TABLE \ref{Table:trials}. Random scheduling explores non-overlapping random actions with the same number of trials as the L-QuasiOpt algorithm and subsequently selects the best cooperative scheduling action from the set of explored actions. As can be seen in the figure and table, the L-QuasiOpt algorithm achieves high ergodic sum-throughput performance even with a very small computational complexity than $2^K$. Further, the L-QuasiOpt algorithm provides considerable gain over Random scheduling.  }

\setlength{\tabcolsep}{3pt}
\begin{table}[!t]
\caption{{Comparison of ergodic sum-throughput [Bits/Sec/Hz] \\ averaged over the number of D2D links $K$.}} \label{Table:EST}
\centering
\begin{tabular}{c|c|c|c}
\hline
 $\text{max}~K_c$ & $8$  &$10$ &$12$ \\ \hline \hline
Optimal & $23.68$ & $23.68$ &$23.68$   \\ \hline
L-QuasiOpt  & $22.65$    & $23.08$    & $23.22$       \\ \hline
Random scheduling  & $20.11$    & $20.92$    & $21.97$        \\ \hline
\end{tabular}
\end{table}

\setlength{\tabcolsep}{3pt}
\begin{table}[!t]
\caption{{ Comparison of trials by the number of links $K$.}} \label{Table:trials}
\centering
\begin{tabular}{c|c|c|c|c|c|c}
\hline
 $K$ & $10$  &$12$ &$14$ &$16$ &$18$ &$20$\\ \hline \hline
Optimal & $2^{10}$ & $2^{12}$ &$2^{14}$ &$2^{16}$ &$2^{18}$ &$2^{20}$  \\ \hline
L-QuasiOpt ($\text{max}~K_c=12$) & $2^{10}$    & $2^{12}$    & $2456$    & $2757$    & $2909$    & $3472$    \\ \hline
L-QuasiOpt ($\text{max}~K_c=10$) & $2^{10}$    & $636$    & $662$    & $681$    & $729$    & $738$    \\ \hline
L-QuasiOpt ($\text{max}~K_c=8$) & $191$    & $225$    & $289$    & $376$    & $382$    & $456$    \\ \hline
\end{tabular}

\end{table}

\begin{figure}[t]
    \centering
    \subfloat[]{\includegraphics[width=8cm]{ 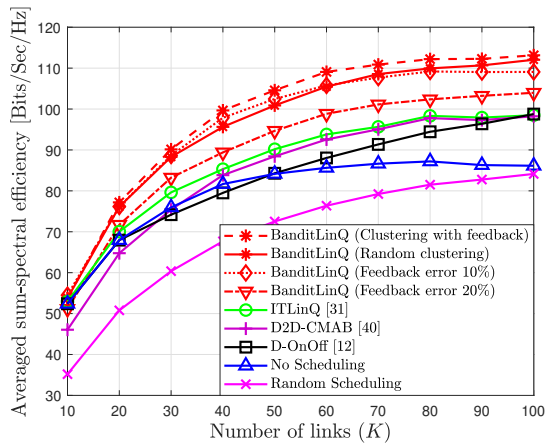}}\\
    \subfloat[]{\includegraphics[width=8cm]{ 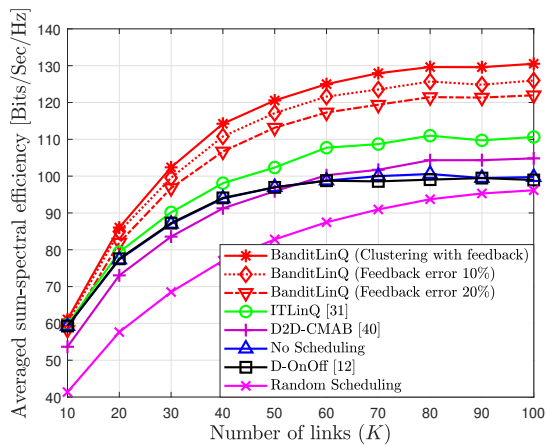}}
    \caption{{ Empirically averaged sum-spectral efficiency $R^{\sf Avg}_{\sf sum}[T]$ comparison of the link scheduling algorithms for (a)  $m=1$ and (b) $m=10$}.}\label{fig:avgST}
\end{figure}

{Fig. \ref{fig:avgST}-(a) and Fig. \ref{fig:avgST}-(b) show the empirically averaged sum-spectral efficiency $R^{\sf Avg}_{\sf sum}[T] = \sum_{t=1}^{T} \frac{ R_{\sf sum}[t](\mathbf{a}_{\pi [t]})}{T}$ for various link scheduling algorithms. We consider the non-line-of-sight (NLOS) scenario for fig. \ref{fig:avgST}-(a) by setting $m=1$, and the line-of-sight (LOS) scenario for fig. \ref{fig:avgST}-(b) by setting $m=10$.}
In this simulation, we drop $K$ D2D links uniformly in a 1 km $\times$ 1 km $\mathbb{R}^2$ plane and set the target rate as $r_k = 3$ for all D2D links. For BanditLinQ and D2D-CMAB algorithms, the BS only uses one-bit feedback information per receiver.  For the D-OnOff algorithm, we assume that each D2D receiver has knowledge about the link quality $|h_{k,k}|^2 d_{k,k}$, distance $d_{k,k}$, and density of a network $\lambda$ at every fading block. In addition, to leverage the threshold defined in \eqref{eq:D-onoff}, we assume that each D2D receiver estimates its own path-loss exponent as $\hat{\beta}_{k,k}$. Since it is infeasible to estimate the path-loss exponent accurately, we assume that the estimation is imperfect. The estimation error is assumed to be distributed as $|\beta_{k,k}-\hat{\beta}_{k,k}| \sim \mathcal{U}(0,0.5)$. Further, for the ITLinQ algorithm, each link is assumed to use its own frequency band thus the receiver perfectly estimates their $\text{SNR}$ and $\text{INR}$ in every fading block. The BanditLinQ (Random clustering) means that the BS divides links randomly and BanditLinQ (Clustering with feedback) means that the BS clusters the link using SNR and INR following a description in \ref{sec:clustering} with {$T_{\sf clust}=10$}. { Additionally, BanditLinQ (Feedback error) takes into account feedback errors in the form of ACK/NACK flip probability, set at $0.1$ and $0.2$.} These assumptions remain consistent unless mentioned otherwise. One remarkable observation is that the proposed BanditLinQ algorithm outperforms the D-OnOff and ITLinQ algorithms in the sense of the averaged sum-spectral efficiency only with one-bit feedback per receiver without any knowledge about the network parameters or interference level. The performance gap between D-OnOff and BanditLinQ is due to the centralized scheduling gain. Moreover, the performance gap between No Scheduling/Random Scheduling and BanditLinQ comes from the one-bit CSI gain. {It is important to note that, despite the degradation in sum-spectral efficiency caused by feedback errors, BanditLinQ still exhibits superior performance compared to other algorithms, especially when with 10 $\%$ feedback errors. This is because introducing a 10 $\%$ perturbation into the utility function has a small impact on the ability to find the quasi-optimal link scheduling action using the BanditLinQ algorithm. Therefore, we conclude that the multi-armed bandit framework employing an objective function based on ACK/NACK feedback is robust to the perturbation of the objective function. }

\begin{figure}[t]
\begin{center}
\includegraphics[width=8cm]{ 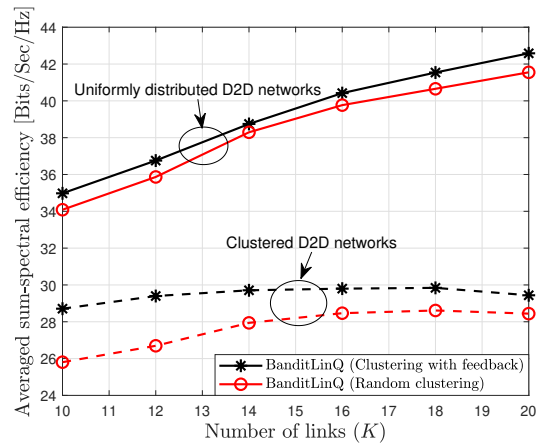}
\end{center}
\caption{Comparison of the BanditLinQ algorithm according to the clustering strategies in different D2D networks for $m = 1$.}\label{fig:clustering}
\end{figure}

\begin{figure}[t]
\begin{center}
\includegraphics[width=8cm]{ 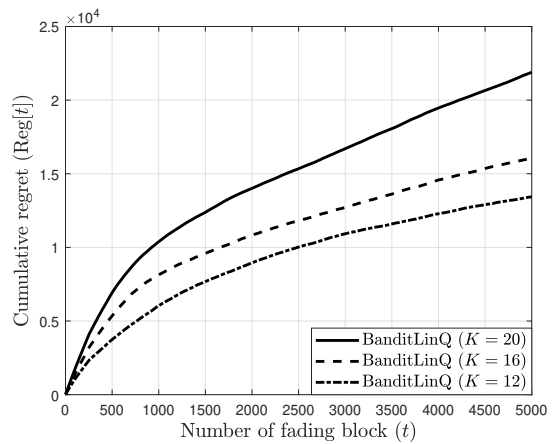}
\end{center}
\caption{{ Cumulative regret performances of BanditLinQ algorithm according to the number of fading blocks for $m=1$.}}\label{fig:regret}
\end{figure}

Another observation is that the link clustering with feedback using INR and SNR shows almost the same sum-spectral efficiency with random clustering. However, when the D2D links are clustered together in multiple clusters in a sparse network, the clustering with feedback reduces the inter-cluster interference significantly and consequently has higher averaged sum-spectral efficiency than random clustering. Fig. \ref{fig:clustering} shows the empirically averaged sum-spectral efficiency of the BanditLinQ algorithm for a sparse network with clustered D2D links in a 0.5 km $\times$ 0.5 km $\mathbb{R}^2$ plane. As can be seen, the performance gain obtained from clustering with feedback is more pronounced in clustered D2D networks compared to uniformly distributed D2D networks.

{Fig. \ref{fig:regret} shows the cumulative regrets ${\sf Reg}[t]$ of the BanditLinQ algorithm according to the number of fading blocks.  We drop $K$ D2D links uniformly in a 0.5 km $\times$ 0.5 km $\mathbb{R}^2$ plane and set the target rate as $r_k$ = 5 for all D2D links. The cumulative regret provides the cumulative ergodic sum-throughput gap between the optimal action and selected action using the BanditLinQ algorithm over fading blocks. The BanditLinQ algorithm shows logarithmic cumulative regret in finding the optimal cooperative link scheduling action during initial fading blocks. This means the ergodic sum-throughput gap between the optimal action and the action chosen by BanditLinQ steadily decreases. However, after a substantial number of fading blocks, the BanditLinQ algorithm shows linear regret, as it converges to a quasi-optimal cooperative link scheduling action. Despite this, the slopes of the regret curve remain small, underscoring the effectiveness of the BanditLinQ algorithm.}

\begin{figure}[t]
\begin{center}
\includegraphics[width=8cm]{ 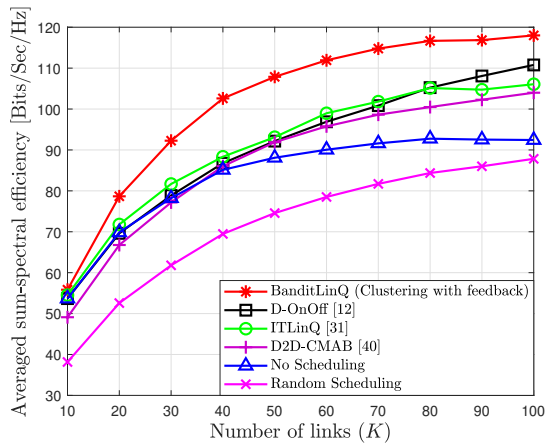}
\end{center}
\caption{{Empirically averaged sum-spectral efficiency $R^{\sf Avg}_{\sf sum}[T]$ comparison considering D2D receivers' mobility for $m=1$.}}\label{fig:AvgSE_mobility}
\end{figure}

{Fig. \ref{fig:AvgSE_mobility} shows empirically averaged sum-spectral efficiency $R^{\sf Avg}_{\sf sum}[T]$ for various link scheduling algorithms considering the device mobility. Specifically, we assume that the D2D receivers' locations change $1$ m randomly from the locations in the previous fading block in every $10$ fading blocks during $T=5000$ fading blocks, while keeping the locations of D2D transmitters fixed. To adapt to the variations in path loss resulting from this mobility, we compute the empirical cluster utility function of the BanditLinQ algorithm by imposing weight parameter $w$ that is less than $1$ to the previously obtained utility function as
\begin{align}
    \hat{\mu}_{j_c,T_{j_c}}[t] &=  \frac{\sum_{u=1}^{t-1}w^{t-u-1}{\bar R}_{\sf sum}[u]({\bf a}_{\pi[u]})\mathbf{1}{\{\pi_c[u] = j_c\}}}{\sum_{s=1}^{t-1}w^{t-s-1}\mathbf{1}{\{\pi_c[s] = j_c\}}}.
\end{align}
In this simulation, we set $w=0.9999$. As can be seen, the BanditLinQ algorithm outperforms other algorithms, especially the D-OnOff algorithm which relies on updated distance information. This highlights the superior ability of the BanditLinQ algorithm to effectively track changes in network topology.}

\section{Conclusion}
In this paper, we considered the D2D link scheduling problem for D2D communications in ultra-dense NB-IoT networks. We first derive the analytic expression of the ergodic sum-spectral efficiency as a function of scheduling action and network parameters. Then, the optimal link scheduling action was obtained by solving the combinatorial optimization problem. To reduce the computational complexity of the combinatorial optimization problem, the L-QuasiOpt algorithm is proposed that clusters the D2D links and selects the cluster link scheduling action that maximizes the cluster utility function. Lastly, we presented the BanditLinQ algorithm which is the link scheduling algorithm using one-bit feedback information. The BanditLinQ algorithm reduces the action space exponentially by clustering the D2D links with one-bit feedback information. Further, by updating the empirical cluster utility function with sum-throughput, the proposed algorithm finds quasi-optimal cooperative link scheduling action. Theoretical analysis proved that the proposed BanditLinQ algorithm achieves the logarithmic regret in finding the quasi-optimal cooperative action which ensures the optimality within a constant sum-throughput gap. Via simulations, we verified that the BanditLinQ algorithm produces a performance gain compared to other link scheduling approaches in the sum-spectral efficiency.

\section{Appendix}
\subsection{Proof for Theorem 1}
\begin{proof}
We compute the ergodic spectral efficiency for cooperative scheduling action $\mathbf{a}_j$ as
\begin{align}
    &R_k(\mathbf{a}_j) \nonumber\\&= \mathbb{E}_{\left|{h}_{k,\ell}[t]\right|^2} \left[ \log_2\left(\!
 1\!+\! \frac{\left|h_{k,k}[t]\right|^2d_{k,k}^{-\beta_{k,k}}{a}_{j,k} }{\sum_{\ell\neq k}\left|{h}_{k,\ell}[t]\right|^2 d_{k,\ell}^{-\beta_{k,\ell}}{a}_{j,\ell} +\frac{1}{\sf snr}} \!\right)\right] \nonumber \\
 & \overset{(a)}{=} \log_2{e} \int_{0}^{\infty}\frac{e^{-\frac{z}{\sf snr}}}{z} \left(1-\mathbb{E}\left[e^{-z |h_{k,k}[t]|^2d_{k,k}^{-\beta_{k,k}}a_{j,k}}\right] \right) \nonumber\\ & ~~~~~~~~~~~~~~~~~~~~~~\cdot \prod_{\ell \ne k} \mathbb{E}\left[e^{-z |h_{k,\ell}[t]|^2d_{k,\ell}^{-\beta_{k,\ell}}a_{j,\ell}}\right] \mathrm{d}z, \label{eq:erSE}
\end{align} 
where (a) is by Lemma 1.

Using the fact that $|h_{k,k}[t]|^2 \sim \text{Gamma}(m,m)$, the first expectation in \eqref{eq:erSE} is computed as
\begin{align}
    \mathbb{E}\left[e^{-z |h_{k,k}[t]|^2d_{k,k}^{-\beta_{k,k}}a_{j,k}}\right]  =\frac{1}{\left(1+\frac{z d_{k,k}^{-\beta_{k,k}}a_{j,k}}{m}\right)^m}. \label{eq:laplace}
\end{align}
 The second expectation in \eqref{eq:erSE} is directly obtained by setting $m=1$ to \eqref{eq:laplace}. Then the ergodic spectral efficiency is expressed as
\begin{align}
    R_k(\mathbf{a}_j)  & = \log_2{e} \int_{0}^{\infty}\frac{e^{-\frac{z}{\sf snr}}}{z} \left(1- {\left(1+\frac{z d_{k,k}^{-\beta_{k,k}}a_{j,k}}{m}\right)^{-m}}\right)  \nonumber\\ &~~~~~~~~\cdot \prod_{\ell \ne k} \frac{1}{1+z d_{k,\ell}^{-\beta_{k,\ell}}a_{j,\ell} } \mathrm{d}z \nonumber \\ 
 &= \log_2{e} \int_{0}^{\infty} a_{j,k} \frac{e^{-\frac{z}{\sf snr}}}{z} \left(1- {\left(1+\frac{z d_{k,k}^{-\beta_{k,k}}}{m}\right)^{-m}}\right)\nonumber\\ &~~~~~~~~~\cdot \prod_{\ell \ne k} \frac{1}{\left(1+z d_{k,\ell}^{-\beta_{k,\ell}}\right)^{a_{j,\ell}} } \mathrm{d}z.
\end{align}
By summing the ergodic spectral efficiencies of all D2D links, we arrive at the expression in \eqref{theorem1}, which completes the proof.
\end{proof}

\subsection{Proof for Theorem 2}
\begin{proof} 
We first prove that the BanditLinQ algorithm achieves optimal logarithmic regret bound in finding the quasi-optimal cooperative action ${\mathbf{a}}_{\hat{j}^\star}$ regardless of the clustering methods. For these derivations, the proof of UCB1 algorithm in \cite{Auer2002} is used as a baseline. 

The cumulative regret is upper bounded by
\begin{align}
    &\mathbb{E}\left[\!{\sf Reg}_{\hat{j}^\star}[n]\!\right] \!=\! n\mu_{\hat{j}^\star}\!-\!\mathbb{E}\left[\!\sum_{t=1}^n \mu _{\pi [t]}\!\right]  \!=\! \mathbb{E}\left[\!\sum_{j=1}^{2^K}T_j[n]\left(\mu_{\hat{j}^\star}\!-\!\mu_j\right)\!\right] \nonumber \\ &= \sum_{\substack{j=1, \\ j\ne \hat{j}^\star}}^{2^K} \Delta_j \mathbb{E}\left[T_j[n]\right] \le \Delta^{\sf max} \sum_{\substack{j=1, \\ j\ne \hat{j}^\star}}^{2^K}  \mathbb{E}\left[T_j[n]\right], \label{eq:regup}
\end{align}
where $\Delta_j = \mu_{\hat{j}^\star}-\mu_j$ and $\Delta^{\sf max}= \underset{j\in [2^K]}{\mathrm{\max}}\mu_{\hat{j}^{\star}}-\mu_{j}$.
The total expected number of selecting the non-optimal cooperative action in \eqref{eq:regup} is upper bounded by the sum of the expected number of selecting the non-optimal action in each cluster as follows:
\begin{align}
    \sum_{\substack{j=1, \\ j\ne \hat{j}^\star}}^{2^K}  \mathbb{E}\left[T_j[n]\right] \le \sum_{c=1}^C \sum_{\substack{j_c=1, \\ j_c\ne \hat{j}_c^\star}}^{2^{K_c}} \mathbb{E}\left[T_{j_c}[n]\right].
\end{align}
Then, the upper bound of cumulative regret is expressed by the expected number of selecting the non-optimal action in each cluster as
\begin{align}
    \mathbb{E}\left[{\sf Reg}_{\hat{j}^\star}[n]\right] \le \Delta^{\sf max} \sum_{c=1}^C \sum_{\substack{j_c=1, \\ j_c\ne \hat{j}_c^\star}}^{2^{K_c}} \mathbb{E}\left[T_{j_c}[n]\right].
\end{align}

Therefore, we obtain the regret upper bound by simply bounding the expected number of selecting the non-optimal action in each cluster, i.e., $\mathbb{E}\left[T_{j_c}[n]\right]$. A non-optimal action $\mathbf{a}_{j_c}^c$ is selected at fading block $t$ if $B_{{j_c},T_{j_c}}[t] \ge B_{{\hat{j}_c^\star},T_{\hat{j}_c^\star}}[t]$.
Then for any integer $u$, the number of selecting the non-optimal action $\mathbf{a}_{j_c}^c$ until fading block $n$ is upper bounded by
\begin{align}
    &T_{j_c}[n]  \nonumber\\
    &\le u +\!\!\! \sum_{t=2^{K_c}+1}^n \!\! \mathbf{1}\Bigl\{B_{\hat{j}_c^\star, T_{\hat{j}_c^\star}}[t\!-\!1] \!\le\! B_{j_c,T_{j_c}}[t\!-\!1], T_{j_c}[t\!-\!1]\!\ge\! u \Bigr\} \nonumber \\
    & \le u + \sum_{t=1}^\infty \sum_{s^\star =1}^{t-1} \sum_{s = u}^{t-1} \mathbf{1}\left\{ B_{\hat{j}_c^\star,s^\star}[t] \le B_{j_c,s}[t]\right\}. \label{eq:ineq1}
\end{align}
The $ B_{\hat{j}_c^\star,s^\star}[t] \le B_{j_c,s}[t]$ in \eqref{eq:ineq1} signifies that at least one of the following must be satisfied:
\begin{align}
    \hat{\mu}_{j_c,s}[t] &\ge \mu_{j_c} \!+\! \sqrt{\frac{2^{K-K_c+1}\left(\sum_{k=1}^K r_k\right)^2 \ln t}{s}}, \label{eq:cond1} \\ \hat{\mu}_{\hat{j}_c^\star,s^\star}[t] &\le \mu_{\hat{j}_c^\star} \!-\! \sqrt{\frac{2^{K-K_c+1}\left(\sum_{k=1}^K r_k\right)^2 \ln t}{s^\star}}, \label{eq:cond2}\\
    \mu_{\hat{j}_c^\star} &< \mu_{j_c}+2\sqrt{\frac{2^{K-K_c+1}\left(\sum_{k=1}^K r_k\right)^2 \ln t}{s}}. \label{eq:cond3}
\end{align}

Let $\Delta_{j_c}= \mu_{\hat{j}_c^{\star}}-\mu_{j_c}$ be the gap between the cluster utility function of quasi-optimal action $\mathbf{a}_{\hat{j}_c^\star}$ and action $\mathbf{a}_{j_c}$ of the cluster $\mathcal{C}_c$.
Then, if we set $u = \left\lceil\frac{2^{K-K_c+3}(\sum_{k=1}^K r_k)^2  \ln n}{\Delta_{j_c}^2}\right\rceil$, the condition \eqref{eq:cond3} is false for $s\ge u$ since
\begin{align}
    &\mu_{\hat{j}_c^\star} - \mu_{j_c}-2\sqrt{\frac{2^{K-K_c+1}\left(\sum_{k=1}^K r_k\right)^2 \ln t}{s}} \nonumber\\&\ge \mu_{\hat{j}_c^\star} - \mu_{j_c} - \Delta_{j_c} =0.
\end{align}

To bound the probability of events in \eqref{eq:cond1} and \eqref{eq:cond2}, we apply Hoeffding's inequality. Considering the cluster $\mathcal{C}_c$, we assume that one action out of the total $2^{K-K_c}$ actions of other clusters is selected with equal probability at each fading block. Therefore, the probability that the difference between empirical cluster utility function $\hat{\mu}_{j_c,s}[t]$ and the true cluster utility function $\mu_{j_c}$ exceeds $\epsilon>0$ is upper bounded by
\begin{align}
    \mathbb{P}\left[|\hat{\mu}_{j_c,s}[t]-\mu_{j_c} |\ge \epsilon \right] &\le 2 \exp \left(-\frac{2s\epsilon^2}{2^{K-K_c} \left(\sum_{k=1}^K r_k\right)^2}\right).  
    \label{eq:Hoe}
\end{align}
Applying $\epsilon = \sqrt{2^{K-K_c+1}\left(\sum_{k=1}^K r_k\right)^2 \ln t / s}$, the \eqref{eq:Hoe} is written as
\begin{align}
    \mathbb{P}\left[\hat{\mu}_{j_c,s}[t]-\mu_{j_c} \ge \epsilon \right] \le t^{-4},
     ~\mathbb{P}\left[\hat{\mu}_{j_c,s}[t]-\mu_{j_c} \le -\epsilon \right] \le t^{-4}.
\end{align}
This bound implies that the empirical cluster utility function will differ from the true cluster utility function by at most  $\epsilon$ with probability at least $1-t^{-4}$. 

Then, by taking expectation at both sides of \eqref{eq:ineq1}, we obtain the upper bound of the expected number of selecting the non-optimal action in cluster $\mathcal{C}_c$ until fading block $n$ as follows:

\begin{align}
     &\mathbb{E}\left[T_{j_c}[n]\right]  
 \!\le \!\left\lceil \! \frac{2^{K-K_c+3}(\sum_{k=1}^K r_k)^2  \ln n}{\Delta_{j_c}^2}\! \right\rceil \!\nonumber\\&+\! \sum_{t=1}^\infty \! \sum_{s^\star}^{t-1}\! \sum_{s = u}^{t-1} \!\left(\! \mathbb{P}{\Big[\hat{\mu}_{j_c,s}[t] \!\ge \mu_{j_c} \!+\! \epsilon \Big]} \!+\! \mathbb{P}{\Big[ \hat{\mu}_{\hat{j}_c^\star,s^\star}[t] \! \le \mu_{\hat{j}_c^\star} \!-\! \epsilon \Big]} \!\right) \nonumber \\  &\le \frac{2^{K-K_c+3}(\sum_{k=1}^K r_k)^2  \ln n}{\Delta_{j_c}^2}  + 1 + \frac{\pi^2}{3}.
\end{align}
Therefore, the proposed BanditLinQ algorithm achieves logarithmic regret in finding the quasi-optimal cooperative action $\mathbf{a}_{\hat{j}^\star}$ as
\begin{align}
    &\mathbb{E}\left[{\sf Reg}_{\hat{j}^\star}[n]\right] \!\le \Delta^{\sf max} \sum_{c=1}^C \sum_{\substack{j_c=1, \\ j_c\ne \hat{j}_c^\star}}^{2^{K_c}} \mathbb{E}\left[T_{j_c}[n]\right] \nonumber\\&\le \Delta^{\sf max} \sum_{c=1}^C \sum_{\substack{j_c=1, \\ j_c\ne \hat{j}_c^\star}}^{2^{K_c}} \left(\!\frac{2^{K-K_c+3}(\sum_{k=1}^K r_k)^2  \ln n}{\Delta_{j_c}^2}  + 1 + \frac{\pi^2}{3}\!\right).
\end{align}

As proved above, the BanditLinQ algorithm finds a quasi-optimal cooperative action ${\mathbf{a}}_{\hat{j}^\star}$ which is the concatenation of quasi-optimal cluster link scheduling actions as
\begin{align}
    {\mathbf{a}}_{\hat{j}^\star} = \left[{\mathbf{a}_{\hat{j}_1^\star}^1}^{\top},{\mathbf{a}_{\hat{j}_2^\star}^2}^{\top},\ldots,{\mathbf{a}_{\hat{j}_C^\star}^C}^{\top}\right]^{\top}\in \{0,1\}^K.
\end{align}
Then, the mean utility function, i.e., ergodic sum-throughput, of quasi-optimal cooperative scheduling action ${\mathbf{a}}_{\hat{j}^\star}$ is written as
\begin{align}
    \mu_{\hat{j}^\star} \!=\! \mathbb{E}_{|h_{k,\ell}[t]|^2}\left[\Bar{R}_{\sf sum}[t](\mathbf{a}_{j}) \Big| j_c = \hat{j}_c^\star,~c=1,2,\ldots C\right].
\end{align}
We assume that $\mu_{\hat{j}_c^\star}<\mu_{\hat{j}^\star}$, and this assumption is realistic since most of the actions of other clusters produce poor sum-throughput performance.

On the other hand, the optimal cooperative scheduling action is expressed as a concatenation of the optimal scheduling action of each cluster as
\begin{align}
    \mathbf{a}_{j^{\star}} = \left[{\mathbf{a}_{j_1^{\star}}^1}^{\top},{\mathbf{a}_{j_2^{\star}}^2}^{\top},\ldots,{\mathbf{a}_{j_C^{\star}}^C}^{\top}\right]^{\top} \in \{0,1\}^K.
\end{align}
Then, the mean utility function of optimal cooperative scheduling action $\mathbf{a}_{j}^\star$ is expressed as
\begin{align}
    \mu_{j^\star} \!=\! \mathbb{E}_{|h_{k,\ell}[t]|^2}\!\!\left[\!\Bar{R}_{\sf sum}[t](\mathbf{a}_{j}) \Big| j_c[t] \!=\! j_c^\star,~c\!=\!1,\ldots, C\right],
\end{align}
which is equal to the sum-throughput for the optimal cluster link scheduling action, i.e., $\mu_{j^\star} = \mu_{{j}_c^\star}$.

Therefore, under the premise that $\mu_{\hat{j}_c^\star}<\mu_{\hat{j}^\star}$, the gap between the ergodic sum-throughput of optimal cooperative action $\mathbf{a}_{j^{\star}}$ and quasi-optimal cooperative action $\mathbf{a}_{\hat{j}^{\star}}$ is upper bounded by
\begin{align}
    \mu_{j^\star} \!-\! \mu_{\hat{j}^\star} & \!\le\! \mu_{j^\star} \!-\! \frac{1}{C}\sum_{c=1}^{C} \mu_{\hat{j}_c^{\star}} \!=\!\frac{1}{C}\sum_{c=1}^{C}\left[\mu_{j_c^{\star}}\!-\!\mu_{\hat{j}_c^{\star}}\right],
\end{align}
which completes the proof.

\end{proof}

\bibliographystyle{IEEEtran}
\bibliography{IEEEabrv,Reference}

\begin{IEEEbiography}[{\includegraphics[width=1in,height=1.25in,clip,keepaspectratio]{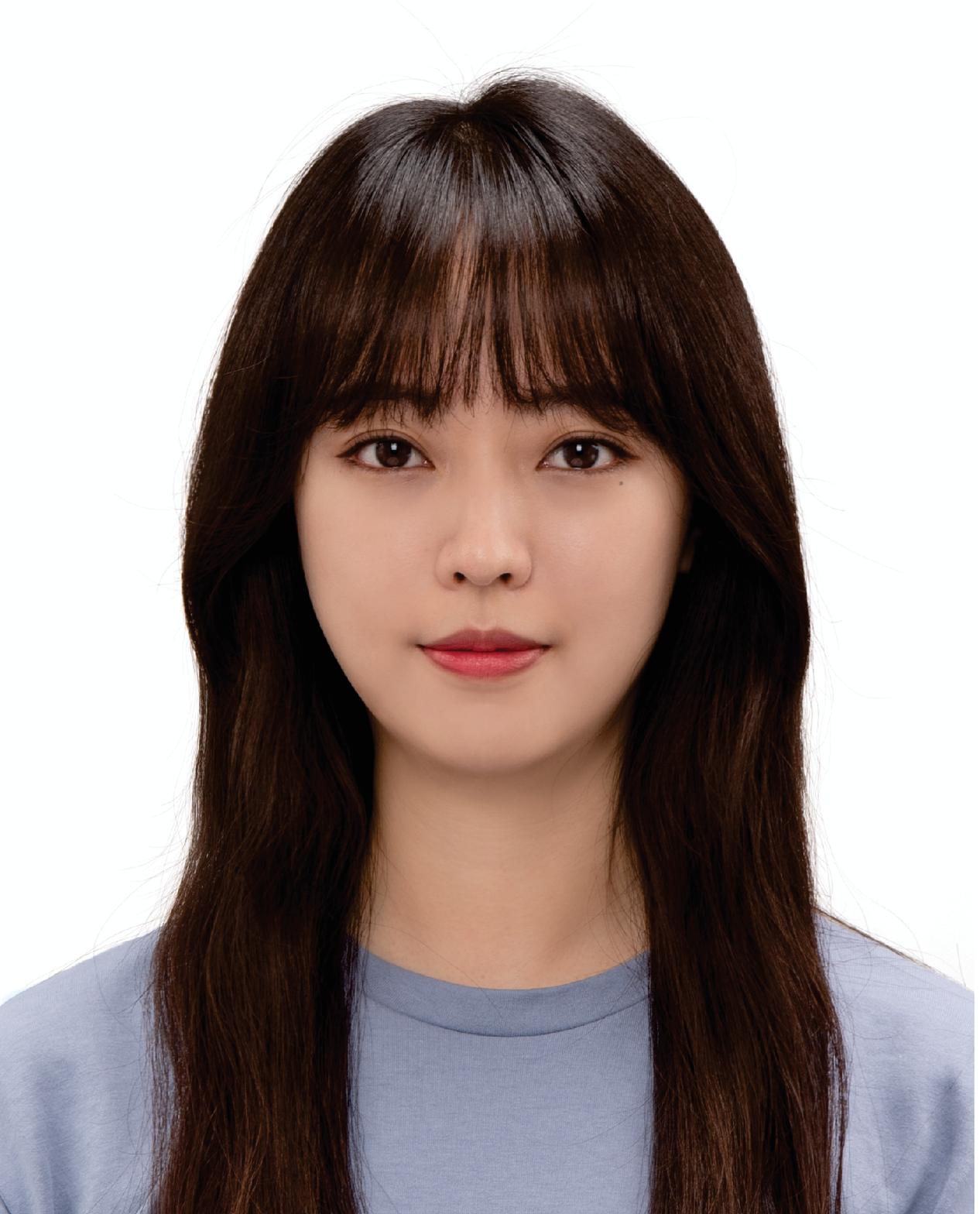}}]{Daeun Kim} (S'20) received the B.S. and M.S. degrees in electrical engineering from Pohang University of Science $\&$ Technology (POSTECH), Pohang, South Korea, in 2018 and 2020, respectively. She is currently pursuing the Ph.D. degree.
\end{IEEEbiography}

\begin{IEEEbiography}[{\includegraphics[width=1in,height=1.25in,clip,keepaspectratio]{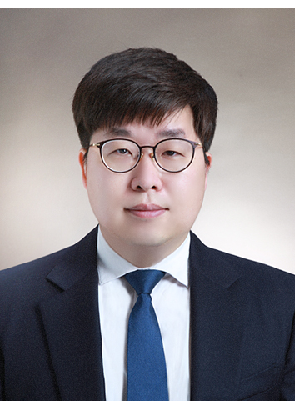}}]{Namyoon Lee} (Senior Member, IEEE) received the Ph.D. degree from The University of Texas at Austin in 2014. He was with the Communications and Network Research Group, Samsung Advanced Institute of Technology, South Korea, from 2008 to 2011, and also with the Wireless Communications Research, Intel Labs, Santa Clara, CA, USA, from 2015 to 2016. He was a Faculty Member with POSTECH from 2016 to 2022, and he is currently an Associate Professor with the School of Electrical Engineering, Korea University. His main research interests are communications, sensing, and machine learning. He was a recipient of the 2016 IEEE ComSoc Asia–Paciﬁc Outstanding Young Researcher Award, the 2020 IEEE Best YP Award (Outstanding Nominee), the 2021 IEEE-IEIE Joint Award for Young Engineer and Scientist, and the 2021 Headong Young Researcher Award. Since 2021, he has been an Associate Editor for IEEE Transactions on Wireless Communications and IEEE Transactions on Communications.
\end{IEEEbiography}

\end{document}